\documentclass[10pt,journal,compsoc]{IEEEtran}
\usepackage{graphicx}

\ifCLASSOPTIONcompsoc
  \usepackage[nocompress]{cite}
\else
  \usepackage{cite}
\fi

\usepackage{booktabs}
\usepackage{smile}
\usepackage{xcolor}
\usepackage{bbm}
\usepackage{pifont}

\usepackage[export]{adjustbox}

\newcommand\ti[1]{\textit{#1}}
\newcommand\tf[1]{\textbf{#1}}
\newcommand\ttt[1]{\texttt{#1}}

\def\ie{\textit{i.e.}}
\def\eg{\textit{e.g.}}

\newcommand{\myparagraph}[1]{\vspace{1pt}\noindent{\bf{#1}}~~}
\makeatletter
\renewcommand{\paragraph}{%
  \@startsection{paragraph}{4}%
  {\z@}{0em}{-1em}%
  {\normalfont\normalsize\bfseries}%
}
\makeatother

\usepackage{colortbl}
\definecolor{lightgray}{gray}{0.75}
\definecolor{lightergray}{gray}{0.85}
\definecolor{Blue}{RGB}{3, 31, 97}
\definecolor{Blue1}{RGB}{214, 235, 245}
\definecolor{Blue2}{RGB}{235, 245, 250}
\definecolor{Gray}{RGB}{247, 252, 255}

\definecolor{convcolor}{HTML}{412F8A}
\definecolor{resnetcolor}{HTML}{8DA0CB}
\definecolor{vitcolor}{HTML}{fc8e62}

\newcommand{\convcolor}[1]{\textcolor{convcolor}{#1}}
\newcommand{\vitcolor}[1]{\textcolor{vitcolor}{#1}}

\definecolor{aliceblue}{rgb}{0.94, 0.97, 1.0}
\newcommand{\vb}{\vitcolor{$\mathbf{\circ}$\,}}
\newcommand{\cb}{\convcolor{$\bullet$\,}}
\newcommand{\gr}{\rowcolor[gray]{.95}}
\newcommand{\cgr}{\cellcolor[gray]{0.95}}

\newcommand{\y}{{\bf y}}
\newcommand{\x}{{\bf x}}
\newcommand{\z}{{\bf z}}

\newcommand{\rr}{{\bf r}}
\newcommand{\uu}{{\bf u}}
\newcommand{\s}{{\bf s}}
\newcommand{\vv}{{\bf v}}
\newcommand{\w}{{\bf w}}

\newcommand{\E}{{\mathcal{E}}}

\usepackage{xspace}
\newcommand{\alg}{\textsc{MONA}\xspace}

\newcommand{\byol}{\textsc{BYOL}\xspace}

\newcommand{\simclr}{\textsc{SimCLR}\xspace}

\newcommand{\moco}{\textsc{MoCo}\xspace}
\newcommand{\mocoo}{\textsc{MoCov2}\xspace}
\newcommand{\mocoknn}{\textsc{$k$NN-MoCo}\xspace}

\newcommand{\isd}{\textsc{ISD}\xspace}
\newcommand{\fcn}{\textsc{FCN}\xspace}
\newcommand{\unet}{\textsc{UNet}\xspace}
\newcommand{\fpn}{\textsc{FPN}\xspace}
\newcommand{\unetF}{\textsc{UNet-F}\xspace}
\newcommand{\unetL}{\textsc{UNet-L}\xspace}
\newcommand{\glc}{\textsc{GLCon}\xspace}
\newcommand{\emm}{\textsc{EM}\xspace}
\newcommand{\cct}{\textsc{CCT}\xspace}
\newcommand{\dan}{\textsc{DAN}\xspace}
\newcommand{\urpc}{\textsc{URPC}\xspace}
\newcommand{\dct}{\textsc{DCT}\xspace}
\newcommand{\ict}{\textsc{ICT}\xspace}
\newcommand{\mt}{\textsc{MT}\xspace}
\newcommand{\uamt}{\textsc{UAMT}\xspace}
\newcommand{\cps}{\textsc{CPS}\xspace}
\newcommand{\scs}{\textsc{SCS}\xspace}
\newcommand{\gcl}{\textsc{GCL}\xspace}
\newcommand{\plc}{\textsc{PLC}\xspace}

\newcommand{\simcvd}{\textsc{SimCVD}\xspace}
\newcommand{\mms}{\textsc{MMS}\xspace}

\usepackage[colorlinks,linkcolor=red, anchorcolor=blue, citecolor=teal, urlcolor=magenta]{hyperref}

\usepackage[capitalize]{cleveref}
\crefname{section}{Sec.}{Secs.}
\Crefname{section}{Section}{Sections}
\Crefname{table}{Table}{Tables}
\crefname{table}{Tab.}{Tabs.}

\makeatletter
\newcommand{\ssymbol}[1]{$^{\@fnsymbol{#1}}$}
\makeatother

\begin{document}

\title{Mine yOur owN Anatomy: \\ Revisiting Medical Image Segmentation with Extremely Limited Labels}

\author{Chenyu~You\ssymbol{1},
        Weicheng~Dai\ssymbol{1},
        Fenglin~Liu,
        Yifei~Min,
        Nicha C. Dvornek,
        Xiaoxiao~Li,
        David~A.~Clifton,
        Lawrence~Staib,
        and~James~S.~Duncan, \IEEEmembership{Life Fellow, IEEE}% <-this % stops a space
\IEEEcompsocitemizethanks{\IEEEcompsocthanksitem Chenyu~You, Weicheng~Dai, Yifei~Min, Lawrence~Staib and James~S.~Duncan are with Yale University, New Haven 06510, U.S.A.. E-mail: \{chenyu.you,weicheng.dai,yifei.min, lawrence.staib, james.duncan\}@yale.edu.
\IEEEcompsocthanksitem Fenglin Liu and David A. Clifton are with University of Oxford, OX3 7DQ Oxford, U.K.. DAC is also with the Oxford-Suzhou Centre for Advanced Research, Suzhou, China. 
E-mail: \{fenglin.liu, david.clifton\}@eng.ox.ac.uk.
\IEEEcompsocthanksitem Xiaoxiao Li is with University of British Columbia, Vancouver, BC V6T 1Z4, Canada.\protect
E-mail: xiaoxiao.li@ece.ubc.ca.
\IEEEcompsocthanksitem \ssymbol{1} Equal Contribution.
}% <-this % stops a space
\thanks{(Corresponding author: Chenyu~You.)}
}

\markboth{IEEE TRANSACTIONS ON PATTERN ANALYSIS AND MACHINE INTELLIGENCE}%
{Shell \MakeLowercase{\textit{et al.}}: Bare Advanced Demo of IEEEtran.cls for IEEE Computer Society Journals}

\IEEEtitleabstractindextext{%
\begin{abstract}
Recent studies on contrastive learning have achieved remarkable performance solely by leveraging few labels in the context of medical image segmentation. Existing methods mainly focus on instance discrimination and invariant mapping (\ie, pulling positive samples closer and negative samples apart in the feature space). However, they face three common pitfalls: 
(1) {{tailness}}: medical image data usually follows an implicit long-tail class distribution. Blindly leveraging all pixels in training hence can lead to the data imbalance issues, and cause deteriorated performance; (2) {{consistency}}: it remains unclear whether a segmentation model has learned meaningful and yet consistent anatomical features due to the intra-class variations between different anatomical features; and (3) {{diversity}}: the intra-slice correlations within the entire dataset have received significantly less attention. This motivates us to seek a principled approach for strategically making use of the dataset itself to discover similar yet distinct samples from {{different anatomical views}}. In this paper, we introduce a novel semi-supervised 2D medical image segmentation framework termed \textbf{M}ine y\tf{O}ur ow\tf{N} \tf{A}natomy (\alg), and make three contributions.
First, prior work argues that every pixel equally matters to the model training; we observe empirically that this alone is unlikely to define meaningful anatomical features, mainly due to lacking the supervision signal. We show two simple solutions towards learning invariances -- through the use of stronger data augmentations and nearest neighbors. Second, we construct a set of objectives that encourage the model to be capable of decomposing medical images into a collection of anatomical features in an unsupervised manner.
Lastly, we both empirically and theoretically, demonstrate the efficacy of our \alg on three benchmark datasets with different labeled settings, achieving new state-of-the-art under different labeled semi-supervised settings. \alg makes minimal assumptions on domain expertise, and hence constitutes a practical and versatile solution in medical image analysis. {We provide the PyTorch-like pseudo-code in supplementary. Codes will be available on {{\href{https://github.com/charlesyou999648/MONA}{here}}}.}
\end{abstract}

}

\maketitle

\IEEEdisplaynontitleabstractindextext

\IEEEpeerreviewmaketitle

\ifCLASSOPTIONcompsoc

\IEEEraisesectionheading{\section{Introduction}\label{section:intro}}
\else
\section{Introduction}
\label{section:intro}
\fi

\IEEEPARstart{W}{ith} 
the advent of deep learning, medical image segmentation has drawn great attention and substantial research efforts in recent years. Traditional supervised training schemes coupled with large-scale annotated data can engender remarkable performance. However, training with massive high-quality annotated data is infeasible in clinical practice since a large amount of expert-annotated medical data often incurs considerable clinical expertise and time. Under such a setting, this poses the question of how models benefit from a large amount of unlabelled data during training. Recently emerged methods based on contrastive learning (CL) significantly reduce the training cost by learning strong visual representations in an unsupervised manner \cite{wu2018unsupervised,oord2018representation,hjelm2018learning,chen2020simple,he2020momentum,grill2020bootstrap,chen2020improved,caron2020unsupervised,quan2022information,li2022mixcl}. A popular way of formulating this idea is through imposing feature consistency to differently augmented views of the same image - which treats each view as an individual instance.

\begin{figure*}[t]
\centering
\includegraphics[width=0.95\linewidth]{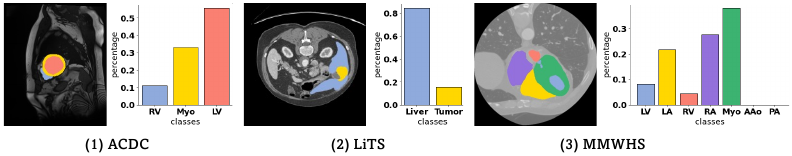}
\vspace{-10pt}
\caption{Examples of three benchmarks (\ie, ACDC, LiTS, MMWHS) with long-tail class distributions. As observed, the ratios of different label classes over three benchmarks are imbalanced.} 
\label{fig:distribution}
\vspace{-10pt}
\end{figure*}

Despite great promise, the main technical challenges remain: (1) How far is CL from becoming a principled framework for medical image segmentation? (2) Is there any better way to implicitly learn some intrinsic properties from the original data (\ie, the inter-instance relationships and intra-instance invariance)? (3) What will happen if models can only access a few labels in training?

% As a core problem in medical image analysis 
To address the above challenges, we outline three principles below: (1) {\em{tailness}}: existing approaches inevitably suffer from class collapse problems -- wherein similar pairs from the same latent class are assumed to have the same representation \cite{arora2019theoretical,chuang2020debiased,li2021prototypical}. This assumption, however, rarely holds for real-world clinical data. 
We observe that the long-tail distribution problem has received increasing attention in the computer vision community \cite{kang2020exploring,zhu2014capturing,cui2019class,yang2020rethinking,jiang2021improving}. 
In contrast, there have been few prior long-tail works for medical image segmentation. 
For example, as illustrated in Figure \ref{fig:distribution}, most medical images follow a Zipf long-tail distribution where various anatomical features share very different class frequencies, which can result in worse performance; (2) {\em{consistency}}: considering the scarcity of medical data in practice, augmentations are a widely adopted pre-text task to learn meaningful representations. 
Intuitively, the anatomical features should be semantically consistent across different transformations and deformations. 
Thus, it is important to assess whether the model is robust to diverse views of anatomy; (3) {\em{diversity}}: recent work \cite{zheng2021ressl,azabou2021mine,van2021revisiting} pointed out that going beyond simple augmentations to create more diverse views can learn more discriminative anatomical features. 
At the same time, this is particularly challenging to both introduce sufficient diversity and preserve the anatomy of the original data, especially in data-scarce clinical scenarios. To deploy into the wild, we need to quantify and address three research gaps from {\em different anatomical views}.

In this paper, we present \textbf{M}ine y\textbf{O}ur ow\textbf{N} \textbf{A}natomy (\alg), a novel contrastive semi-supervised 2D medical segmentation framework, based on different anatomical views. The workflow of \alg is illustrated in Figure~\ref{fig:framework}. The \textbf{key innovation} in \alg is to seek diverse views (\ie, augmented/mined views) of different samples whose anatomical features are {\em{homogeneous}} within the {\em{same class type}}, while {\em{distinctive}} for {\textit{different class types}}. We make the following contributions. First, we consider the problem of {\em{tailness}}. An issue is that label classes within medical images typically exhibit a long-tail distribution. Another one, technically more challenging, is the fact that there is only a few labeled data and large quantities of unlabeled ones during training. Intuitively we would like to sample more pixel-level representations from tail classes. Thus, we go beyond the na\"ive setting of instance discrimination in CL \cite{chen2020simple,he2020momentum,grill2020bootstrap} by decomposing images into diverse and yet consistent anatomical features, each belonging to different classes. In particular, we propose to use pseudo labeling and knowledge distillation to learn better pixel-level representations within multiple semantic classes in a training mini-batch. Considering performing pixel-level CL with medical images is impractical for both memory cost and training time, we then adopt active sampling strategies \cite{liu2021bootstrapping} such as in-batch hard negative pixels, to better discriminate the representations at a larger scale.

We further address the two other challenges: {\em consistency} and {\em diversity}. The success of the common CL theme is mainly attributed to invariant mapping \cite{hadsell2006dimensionality} and instance discrimination \cite{wu2018unsupervised,chen2020simple}. Starting from these two key aspects, we try to further improve the segmentation quality. More specifically, we suggest that {\em{consistency}} to transformation (equivariance) is an effective strategy to establish the invariances (\ie, anatomical features and shape variance) to various image transformations. Furthermore, we investigate two ways to include diversity-promoting views in sample generation. 
First, we incorporate a memory buffer to alleviate the demand for large batch size, enabling much more efficient training without inhibiting segmentation quality. Second, we leverage stronger augmentations and nearest neighbors to mine views as positive views for more semantic similar contexts.

\begin{figure*}[t]
\centering
\includegraphics[width=0.92\linewidth]{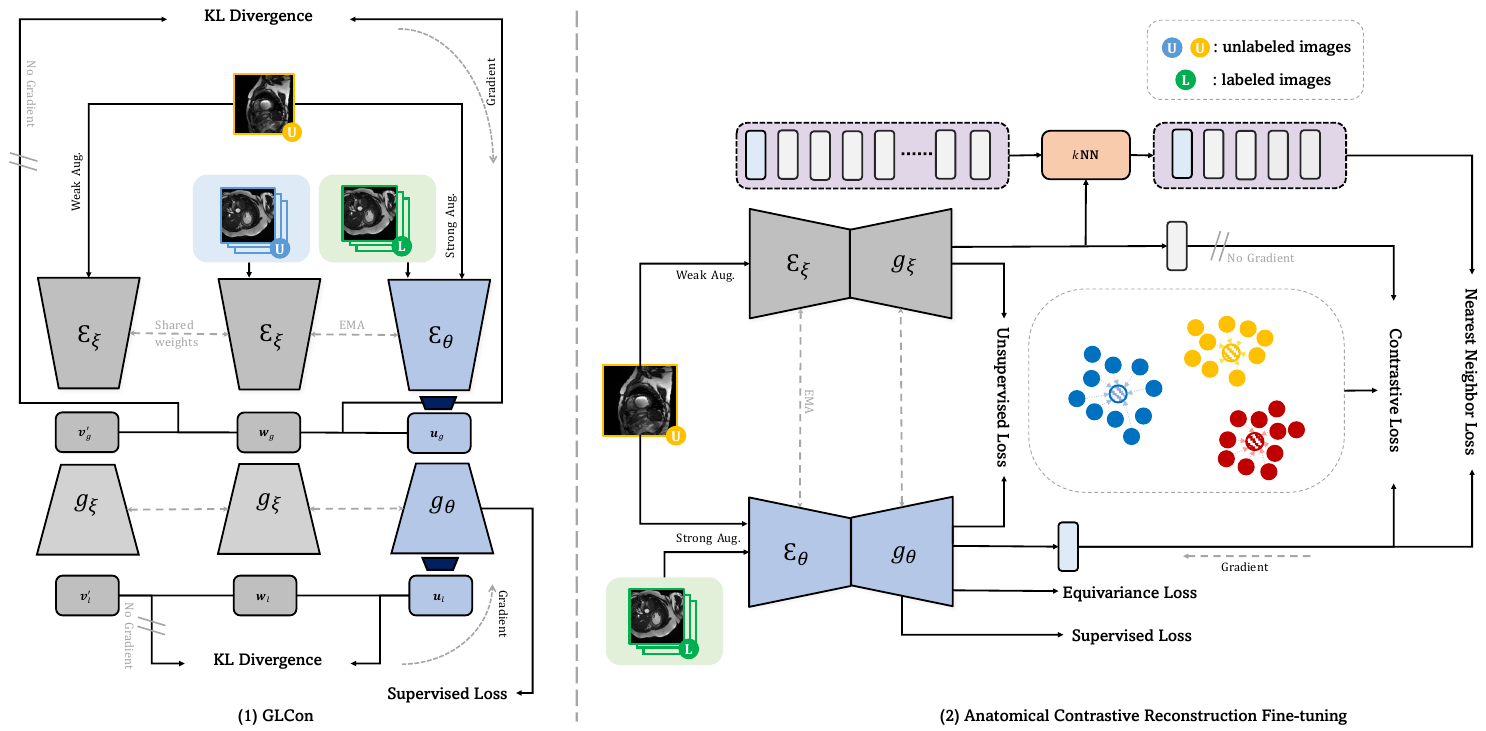}
\vspace{-5pt}
\caption{Overview of the \alg framework including two stages: (1){GLCon is design to seek both \textit{augmented} and \textit{mined} views for instance discrimination $\mathcal{L}_{\text{inst}}$ in the global and local manners. Here the global instance discrimination is designed to exploit the correlations among views within the latent feature space, which is generated by the encoders. Meanwhile, local instance discrimination aims to leverage the correlations among views - specifically, local regions of the image - within the output feature space produced by the decoder (See Section \ref{subsection:framework})}, (2) our proposed anatomical contrastive reconstruction fine-tuning {(See Section \ref{subsection:acr})}. Note that \textrm{U} and \textrm{L} denote unlabeled and labeled data.} 
\label{fig:framework}
\vspace{-10pt}
\end{figure*}

Extensive experiments are conducted on a variety of datasets and the latest CL frameworks (\ie, \moco\cite{he2020momentum}, \simclr\cite{chen2020simple}, \byol\cite{grill2020bootstrap}, and \isd\cite{tejankar2021isd}), which consistently demonstrate the effectiveness of our proposed \alg. For example, our \alg establishes the \textbf{new state-of-the-art} performance, compared to all the state-of-the-art semi-supervised approaches with different label ratios (\ie, 1\%, 5\%, 10\%). Moreover, we present a systematic evaluation for analyzing why our approach performs so well and how different factors contribute to the final performance (See Section \ref{subsection:ablation}).
Theoretically, we show the efficacy of our MONA in label efficiency (See Section~\ref{section:ablation-theory}).
Empirically, we also study whether these principles can effectively complement each CL framework (See Section~\ref{section:ablation-CL}). We hope our findings will provide useful insights on medical image segmentation to other researchers.

To summarise, our contributions are as follows: \ding{182} we carefully examine the problem of semi-supervised 2D medical image segmentation with extremely limited labels, and identify the three principles to address such challenging tasks; \ding{183} we construct a set of objectives, which significantly improves the segmentation quality, both long-tail class distribution and anatomical features; \ding{184} we both empirically and theoretically analyze several critical components of our method and conduct thorough ablation studies to validate their necessity; \ding{185} with the combination of different components, we establish state-of-the-art under SSL settings, for all the challenging three benchmarks.

\section{Related work}
\label{section:related}
\myparagraph{Medical Image Segmentation.}
Medical image segmentation aims to assign a class label to each pixel in an image, and plays a major role in real-world applications, such as assisting the radiologists for better disease diagnosis and reduced cost. With sufficient annotated training data, significant progress has been achieved with the introduction of Fully convolutional networks (\fcn) \cite{long2015fully} and \unet \cite{ronneberger2015u}. Follow-up works can be categorized into two main directions. One direction is to improve modern segmentation network design. Many CNN-based \cite{simonyan2014very,he2016deep} and Transformer-like \cite{vaswani2017attention,dosovitskiy2020image} model variants \cite{milletari2016v,chen2017deeplab,alom2018recurrent,oktay2018attention,chen2018encoder,chen2021transunet,cao2021swin,xie2021cotr,hatamizadeh2021unetr,valanarasu2021medical,you2022class} have been proposed since then. For example, some works \cite{chen2017deeplab,chen2018encoder,dai2017deformable} proposed to use dilated/atrous/deformable convolutions with larger receptive fields for more dense anatomical features. Other works \cite{chen2021transunet,cao2021swin,xie2021cotr,hatamizadeh2021unetr,valanarasu2021medical,you2022class} include Transformer blocks to capture more long-range information, achieving the impressive performance. A parallel direction is to select proper optimization strategies, by designing loss functions to learn meaningful representations \cite{lin2017focal,xue2019shape,shi2021marginal}. However, those methods assume access to a large, labeled dataset. This restrictive assumption makes it challenging to deploy in most real-world clinical practices. In contrast, our \alg is more robust as it leverages only a few labeled data and large quantities of unlabeled one in the learning stage.

\myparagraph{Semi-Supervised Learning (SSL).}
{The goal in robust SSL is to improve the medical segmentation performance by taking advantage of large amounts of unlabelled data during training. It can be roughly categorized into three groups: (1) self-training by generating unreliable pseudo-labels for performance gains, such as pseudo-label estimation \cite{lee2013pseudo,bai2017semi,fan2020inf,chen2021semi,nassar2023protocon}, model uncertainty \cite{yu2019uncertainty,graham2019mild,cao2020uncertainty}, confidence estimation \cite{blundell2015weight,gal2016dropout,kendall2017uncertainties}, and noisy student \cite{xie2020self};} (2) consistency regularization \cite{bortsova2019semi,cui2019semi,fotedar2020extreme} by integrating consistency corresponding to different transformation, such as pi-model \cite{sajjadi2016regularization}, co-training \cite{qiao2018deep,zhou2019semi}, and mean-teacher \cite{tarvainen2017mean,quan2022information,li2022mixcl,zhao2022meta,li2020transformation,reiss2021every}; (3) other training strategies such as adversarial training \cite{zhang2017deep,nie2018asdnet,zhang2018translating,zheng2019semi,li2020shape,valvano2021learning} and entropy minimization \cite{grandvalet2004semi}.  In contrast to these works, we do not explore more advanced pseudo-labelling strategy to learn spatially structured representations. In this work, we are the first to explore a novel direction for discovering distinctive and semantically consistent anatomical features without image-level or region-level labels. Further, we expect that our findings can be relevant for other medical image segmentation frameworks.

\myparagraph{Contrastive Learning.}
{CL has recently emerged as a promising paradigm for medical image segmentation via exploiting abundant unlabeled data, leading to state-of-the-art results \cite{chaitanya2020contrastive,dwibedi2021little,you2021momentum,chaitanya2021local,hu2021semi,li2022mixcl,quan2022information,you2022simcvd,zhang2023multi,lou2023min}. The high-level idea of CL is to pull closer the different augmented views of the same instance but pushes apart all the other instances away. Intuitively, differently augmented views of the same image are considered {\em positives}, while all the other images serve as {\em negatives}. 
The major difference between different CL-based frameworks lies in the augmentation strategies to obtain {\em positives} and {\em negatives}. 
\cite{dangovski2022equivariant} augments a given image with 4 different rotation degrees and trains the model to be aware of which rotation degree of each image by applying an contrastive loss. In contrast, our goal is to train a model to yield segments that adhere to anatomical, geometric and equivariance constraints in an unsupervised manner.}
A few very recent studies \cite{kang2020exploring,jiang2021improving} confirm the superiority of CL of addressing imbalance issues in image classification. Moreover, existing CL frameworks \cite{chaitanya2020contrastive,you2021momentum} mainly focus on the instance level discrimination (\ie, different augmented views of the same instance should have similar anatomical features or clustered around the class weights). 
However, we argue that not all negative samples equally matter, and the above issues have not been explored from the perspective of medical image segmentation, considering the class distributions in the medical image are perspectives diverse and always exhibit long tails \cite{galdran2021balanced,yan2022sam,roy2022does}. 
Inspired by the aforementioned, we address these two issues in medical image segmentation - two appealing perspectives that still remain under-explored.
\section{Mine yOur owN Anatomy (\alg)}
\label{section:method}

\myparagraph{Overview.}
\alg consists of two parts: {a global-local contrastive pre-training part named GLCon (Section~\ref{subsection:framework}) and a fine-tuning part named Anatomical Contrastive Reconstruction (Section~\ref{subsection:acr}).}
We illustrate our contrastive learning framework (See Figure~\ref{fig:framework}), which includes (1) relational semi-supervised pre-training, and (2) anatomical contrastive reconstruction fine-tuning.

\subsection{{GLCon}}
\label{subsection:framework}

Our pre-training stage is built upon \isd \cite{tejankar2021isd} - a competitive framework for image classification. 
The {\em{main differences}} between {\isd} and {the pre-training part of {\alg} (\ie~GLCon)} are: {GLCon is more tailored to medical image segmentation, \ie, considering the dense nature of this problem both in global and local manner, and can generalize well to those long-tail scenarios.} 
Also, our principles are expected to apply to other CL framework ((\ie, \moco\cite{he2020momentum}, \simclr\cite{chen2020simple}, \byol\cite{grill2020bootstrap}). More detailed empirical and theoretical analysis can be found in Section~\ref{section:ablation-CL} and Section~\ref{section:ablation-theory}.

\myparagraph{Pre-training preliminary.}
Let $(X, Y)$ be our dataset, including training images $\x \in X$ and their corresponding $\mathcal{C}$-class segmentation labels  $\y\in Y$, where $X$ is composed of $N$ labeled and $M$ unlabeled slices. Note that, for brevity, $\y$ can be either sampled from $Y$ or pseudo-labels. The {\em{student}} and {\em{teacher}} networks $\mathcal{F}$, parameterized by weights $\theta$ and  $\xi$, each consist of a encoder $\mathcal{E}$ and a decoder $\mathcal{D}$ (\ie, \ttt{UNet} \cite{ronneberger2015u}). Concretely, given a sample $\s$ from our unlabeled dataset, we have two ways to generate views: (1) we formulate {\em{augmented}} views (\ie, $\x,\x'$) through two different augmentation chains; and (2) we create $d$ {\em{mined}} views (\ie, $\x_{r,i}$) by randomly selecting from the unlabeled dataset followed by additional augmentation.\footnote{Note that the subscript $i$ is omitted for simplicity in following contexts.}
We then fed the {\em{augmented}} views to both $\mathcal{F}_\theta$ and  $\mathcal{F}_\xi$, and the {\em{mined}} views to $\mathcal{F}_\xi$. Similar to \cite{chaitanya2020contrastive}, we adopt the global and local instance discrimination strategies in the latent and output feature spaces.\footnote{Here we omit details of local instance discrimination strategy for simplicity because the global and local instance discrimination experimental setups are similar.}
Specifically, the encoders generate global features $\z_{g} = \E_\theta( \x)$, $\z'_{g} = \E_\xi( \x')$, and $\z_{r,g} = \E_\xi( \x_r)$, which are then fed into the nonlinear projection heads to obtain $\vv_{g} = h_\theta(\z_{g})$, $\vv'_{g} = h_{\xi} (\z'_{g})$, and $\w_{g} = h_{\xi}(\z_{r,g})$. The {\em{augmented}} embeddings from the {\em{student}} network are further projected into secondary space, \ie, ${\uu_{g}} = h'_\theta(\vv_{g})$. We calculate similarities across {\em{mined}} views and {\em{augmented}} views from the {\em{student}} and {\em{teacher}} in both global and local manners. Then a \texttt{softmax} function is applied to process the calculated similarities, which models the relationship distributions:
\begin{equation}
    \begin{aligned}
    \s_\theta &= \text{log} \frac{\text{exp}\big(\text{sim}\big(\uu, \w\big)/\tau_\theta\big)}{\sum_{j=1}^k \text{exp}\big(\text{sim}\big(\uu, \w_{j}\big)/\tau_\theta\big)},\quad \\
    \s_\xi &= \text{log} \frac{\text{exp}\big(\text{sim}\big(\vv', \w\big)/\tau_\xi\big)}{\sum_{j=1}^k \text{exp}\big(\text{sim}\big(\vv', \w_{j}\big)/\tau_\xi\big)},
\end{aligned}
\end{equation}
where $\tau_\theta$ and $\tau_\xi$ are different temperature parameters, {$k$ denotes the number of mined views} and $\text{sim}(\cdot,\cdot)$ denotes cosine similarity. The unsupervised instance discrimination loss (\ie, Kullback-Leibler divergence $\mathcal{KL}$) can be defined as:
\begin{equation}
\mathcal{L}_{\text{inst}} = \mathcal{KL}(\s_\theta || \s_\xi).
\label{equation:pcl}
\end{equation}
The parameters $\xi$ of $\mathcal{F}_{\xi}$ is updated as: $\xi = t \xi + (1-t) \theta$ with $t=0.99$ as a momentum hyperparameter. 
In our pre-training stage, the total loss is the sum of global and local instance discrimination loss $\mathcal{L}_\text{inst}$ (on pseudo-labels), and supervised segmentation loss $\mathcal{L}_\text{sup}$ (\ie, equal combination of dice loss and cross-entropy loss on ground-truth labels): $\mathcal{L^{\text{global}}_\text{inst}} + \mathcal{L^{\text{local}}_\text{inst}} + \mathcal{L}_\text{sup}$. {Therefore, the GLCon loss encourages that the model acquires both global and local features.}
% $H(\mathbf{p}^1,\mathbf{p}^2)$

% {\bf Principles.}

% Section~\ref{section:ablation-theory}.

% address the inequality of negative samples in CL, and the multi-class imbalanced medical image segmentation; consistency ensures the feature invariances; and diversity further encourages to discover more anatomical features in different images.
\begin{figure}[t]
    \centering
    % \vspace{-1em}
    \includegraphics[width=0.65\linewidth]{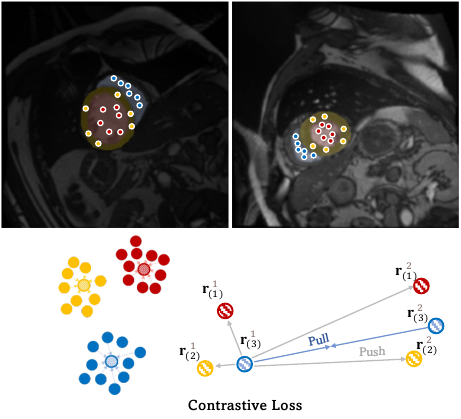}
    \vspace{-10pt}
    \caption{Illustration of the contrastive loss. Intuitively, we actively sample a set of pixel-level anchor representations, pulling them closer to the class-averaged mean of representations within this class ({\ti{positive keys}}), and pushing away from representations from other classes ({\ti{negative keys}}).}
    \label{fig:contrastive_loss}
    \vspace{-5pt}
\end{figure}

% , parallel to the segmentation head, \cite{liu2021bootstrapping,wang2022semi}
\subsection{Anatomical Contrastive Reconstruction}
\label{subsection:acr}
\myparagraph{Principles.}
{The key idea of the fine-tuning part is to seek diverse yet semantically consistent views whose anatomical features are {\em{homogeneous}} within the {\em{same class type}}, while {\em{distinctive}} for {\textit{different class types}}.
As shown in Figure~\ref{fig:framework}, the principles behind \alg (the anatomical contrastive reconstruction stage) aim to ensure tailness, consistency, and diversity. 
Concretely, tailness is for actively sampling more tail class hard pixels; consistency ensures the feature invariances; and diversity further encourages to discover more anatomical features in different images. More theoretical analysis is in Section~\ref{section:ablation-theory}.
}

\myparagraph{Tailness.}
Motivated by the observations (Figure~\ref{fig:distribution}), our primary cue is that medical images naturally exhibit an imbalanced or long-tailed class distribution, wherein many class labels are associated with only a few pixels. To generalize well on such {\em{imbalanced}} setting, we propose to use {\em{anatomical contrastive formulation}} (\textbf{ACF}) (See Figure~\ref{fig:contrastive_loss}).

Here we additionally attach the representation heads to fuse the multi-scale features with the feature pyramid network (\fpn) \cite{lin2017feature} structure and generate the $m$-dimensional representations with consecutive convolutional layers. The high-level idea is that the features should be very {\em{similar}} among the same class type, while very {\em{dissimilar}} across different class types. Particularly for long-tail medical data, a na\"ive application of this idea would require substantially computational resources proportional to the square of the number of pixels within the dataset, and naturally overemphasize the anatomy-rich head classes and leaves the tail classes under-learned in learning invariances, both of which suffer performance drops.

To this end, we address this issue by actively sampling a set of pixel-level anchor representations $\rr_q\in \mathcal{R}_q^c$ ({\em{queries}}), pulling them closer to the class-averaged mean of representations $\rr_k^{c,+}$ within this class $c$ ({\em{positive keys}}), and pushing away from representations $\rr_k^{-}\in \mathcal{R}_k^c$ from other classes ({\em{negative keys}}). Formally, the contrastive loss is defined as:
\begin{equation}
  \label{loss:contrastive}
  \begin{aligned}
         \mathcal{L}_\text{contrast} &= \sum_{c\in \mathcal{C}} \sum_{\rr_q \sim \mathcal{R}^c_q} \\ 
         & -\log \frac{\exp(\rr_q \cdot \rr_k^{c, +} / \tau)}{\exp(\rr_q \cdot \rr_k^{c, +}/ \tau) + \sum_{\rr_k^{-}\sim \mathcal{R}^c_k} \exp(\rr_q \cdot \rr_k^{-}/ \tau)},
  \end{aligned}
\end{equation}
where $\mathcal{C}$ denotes a set of all available classes for each mini-batch, and $\tau$ is a temperature hyperparameter. 
Suppose $\mathcal{A}$ is a collection including all pixel coordinates within $\x$, these representations are:
\begin{equation}
  \begin{aligned}
  \mathcal{R}_q^c &= \bigcup_{[m, n]\in \mathcal{A}}\!\!\mathbbm{1}(\y_{[m,n]}\!=\!c)\, \rr_{[m,n]},\, \\ 
  \mathcal{R}_k^c &= \bigcup_{[m, n]\in \mathcal{A}}\!\!\mathbbm{1}(\y_{[m,n]}\!\neq\!c)\, \rr_{[m,n]},\, \\
  \rr_k^{c, +} &= \frac{1}{| \mathcal{R}_q^c |}\sum_{\rr_q \in \mathcal{R}_q^c} \rr_q\,.
  \end{aligned}
\end{equation}
Note that in Eq.~\ref{loss:contrastive}, we are using the negative pairs $\rr_k^{-}$ to estimate the centers of opposite classes. 
The {class average representation} $\rr_k^{c,+}$ is averaged over all instances from the target class $c$.
We also note that CL might benefit more, where the instance discrimination task is achieved by incorporating more positive and negative pairs. 
However, naively unrolling CL to this setting is impractical since it requires extra memory overheads that grow proportionally with the amount of instance discrimination tasks. To this end, we adopt a random set (\ie, the mini-batch) of other images. 
Intuitively, we would like to maximize the anatomical similarity between all the representations from the query class, and analogously minimize all other class representations. 
{In order to compare the pairs of instances} between opposite and target classes, we then create a graph $\mathcal{G}$ to compute the pair-wise class relationship:
$\mathcal{G}[p, q] = \left(\rr_k^{p, +} \cdot \rr_k^{q, +}\right), \forall p,q \in \mathcal{C}, \text{ and } p\neq q,$
where $\mathcal{G}\in \mathbb{R}^{|\mathcal{C}| \times |\mathcal{C}|}$. 
Here finding the accurate decision boundary can be formulated mathematically by normalizing the pair-wise relationships among all negative class representations via the \texttt{softmax} operator. 
{To be specific, in Eq.~\ref{loss:contrastive}, we use adaptive sampling for the negative keys $\rr_k^{-}$ from the opposite classes. To do so, we use \texttt{softmax} to yield a distribution $\exp(G[c, v])/ \sum_{n\in \mathcal{C}, n\neq c} \exp(G[c, n])$, with which we adaptively sample negative keys from class $v$, for $v\neq c$.} 
To address the challenge in \textit{imbalanced} medical image data, we define the pseudo-label (\ie, easy and hard queries) based on a defined threshold as follows:
\begin{equation}
  \label{eq:easyhard}
  \begin{split}
  \mathcal{R}_q^{c,\, \text{easy}} &= \bigcup_{\rr_q \in \mathcal{R}^c_q} \mathbbm{1}(\hat{\y}_q > \delta_\theta)\rr_q,\quad \\
  \mathcal{R}_q^{c,\, \text{hard}} &= \bigcup_{\rr_q \in \mathcal{R}^c_q} \mathbbm{1}(\hat{\y}_q \leq \delta_\theta)\rr_q,      
  \end{split}
\end{equation}
where $\hat{\y}_q$ is the $c^\text{th}$-class pseudo-label corresponding to $\rr_q$, and $\delta_\theta$ is the user-defined threshold. For further improvement in long-tail scenarios, we construct a class-aware memory bank \cite{he2020momentum} to store a fixed number of negative samples per class $c$.

\begin{figure}[t]
    \centering
    \includegraphics[width=0.65\linewidth]{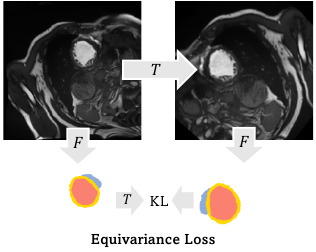}
    \vspace{-5pt}
    \caption{Illustration of the equivariance loss.}
    \vspace{-5pt}
    \label{fig:eqv_loss}
\end{figure}

\myparagraph{Consistency.}
The proposed ACF is designed to address {\em{imbalanced}} issues, but {\em{anatomical consistency}} remains to be weak in the long-tail medical image setting since medical segmentation should be robust to different tissue types which show different anatomical variations. 
{Our goal is to train a model to yield segments that adhere to anatomical, geometric and equivariance constraints in an unsupervised manner.}
As shown in Figure \ref{fig:eqv_loss}, we hence construct a random image transformation $\mathcal{T}$ and define the equivariance loss on both labeled and unlabeled data by measuring the feature consistency distance between each original segmentation map and the segmentation map generated from the transformed image:
\begin{equation}
  \label{eq:eqv}
  \begin{aligned}
    \mathcal{L}_\text{eqv}(\x, \mathcal{T}(\x)) &= \sum_{\x\in X} \mathcal{KL}\left( \mathcal{T}(\mathcal{F}_{\theta}(\x)), \mathcal{F}_{\theta}(\mathcal{T}(\x)) \right) \\
    &+ \mathcal{KL}\left( \mathcal{F}_{\theta}(\mathcal{T}(\x)), \mathcal{T}(\mathcal{F}_{\theta}(\x)) \right)\,.
  \end{aligned}
\end{equation}
Here we define $\mathcal{T}$ on both the input image $\x$ and $\mathcal{F}_{\theta}(\x)$, via the random transformations (\ie, affine, intensity, and photo-metric augmentations), since the model should learn to be robust and invariant to these transformations.

\myparagraph{Diversity.}
Oversampling too many images from the random set would create extra memory overhead, and more importantly, our finding also uncovers that a large number of random images might not necessarily help impose additional invariances between neighboring samples since redundant images might introduce additional noise during training (see Section \ref{section:ablation-acr}). 
{To counteract this, we utilize the nearest neighbor strategy, ensuring the model benefits from its previous outputs  without overly concentrating on extraneous features.
Thus, we formulate our insight as an auxiliary loss that regularizes the representations - keeping the anatomical contrastive reconstruction task as the main force. In practice, given a batch of unlabeled images, we use both the teacher and student models to obtain $\vv'_{g}$ and ${\uu_{g}}$, which are then normalized using the $l_{2}$ norm. $\vv'_{g}$ is fed to the first-in-first-out (FIFO) memory bank \cite{he2020momentum}, where it search for $K$-nearest neighbors from the memory bank. Then we use the nearest neighbor loss $\mathcal{L}_\text{nn}$ to maximize cosine similarity, thereby exploiting the inter-instance relationship. Specifically, we minimize the distance between ${\uu_{g}}$ and the $K$-nearest neighbors, with the distance defined as negative cosine similarity, thereby maximizing cosine similarity.}

\myparagraph{Setup.}
The total loss $\mathcal{L}_\text{total}$ is the sum of contrastive loss $\mathcal{L}_\text{contrast}$ (on both ground-truth labels and pseudo-labels), equivariance loss $\mathcal{L}_\text{eqv}$ (on both ground-truth labels and pseudo-labels), nearest neighbors loss $\mathcal{L}_\text{nn}$ (on both ground-truth labels and pseudo-labels), unsupervised cross-entropy loss $\mathcal{L}_\text{unsup}$ (on pseudo-labels) and supervised segmentation loss  $\mathcal{L}_\text{sup}$ (on ground-truth labels): $\mathcal{L}_\text{sup} + \lambda_{1}\mathcal{L_\text{contrast}} + \lambda_{2}\mathcal{L}_\text{eqv} + \lambda_3\mathcal{L}_\text{unsup} + \lambda_4\mathcal{L}_\text{nn}$. 
We theoretically analyze the effectiveness of our MONA in the very limited label setting (See Section~\ref{section:ablation-theory}). We also empirically conduct ablations on  different hyperparameters (See Section~\ref{section:ablation-acr}).

\section{Experiments}
\label{section:experiments}

\begin{table*}[t]
	\begin{center}
	\caption{Comparison of segmentation performance (DSC{[}\%{]}/ASD{[}mm{]}) on ACDC and LiTS under three labeled ratio settings (1\%, 5\%, 10\%). The best results are indicated in \tf{bold}.}
	\vspace{-5pt}
	\label{table:acdc_lits_main}
    \begin{adjustbox}{width=0.9\linewidth}
	\begin{tabular}{ccccccccccccc}
		\toprule
		& \multicolumn{6}{c}{\textbf{ACDC}} & \multicolumn{6}{c}{\textbf{LiTS}}\\
		\cmidrule(r){2-7} \cmidrule(r){8-13}
		& \multicolumn{2}{c}{1\% Labeled} & \multicolumn{2}{c}{5\% Labeled} & \multicolumn{2}{c}{10\% Labeled} & \multicolumn{2}{c}{1\% Labeled} & \multicolumn{2}{c}{5\% Labeled} & \multicolumn{2}{c}{10\% Labeled}\\
        \cmidrule(r){2-3} \cmidrule(r){4-5} \cmidrule(r){6-7} \cmidrule(r){8-9} \cmidrule(r){10-11} \cmidrule(r){12-13}
		{Method}
		            & {DSC~$\uparrow$}
		            & {ASD~$\downarrow$}
		            & {DSC~$\uparrow$}
		            & {ASD~$\downarrow$}  
		            & {DSC~$\uparrow$}
		            & {ASD~$\downarrow$} 
		            & {DSC~$\uparrow$}
		            & {ASD~$\downarrow$}
		            & {DSC~$\uparrow$}
		            & {ASD~$\downarrow$}  
		            & {DSC~$\uparrow$}
		            & {ASD~$\downarrow$} 
		            \\ \midrule
	    \unetF \cite{ronneberger2015u}
		            & {91.5} 
		            & {0.996}
		            & {91.5} 
		            & {0.996}
		            & {91.5} 
		            & {0.996}
		            & {68.5}
		            & {17.8}
		            & {68.5}
		            & {17.8}
		            & {68.5}
		            & {17.8}
                    \\
		\unetL        
                    & {14.5}
                    & {19.3}
                    & {51.7}
		            & {13.1}
                    & {79.5}
                    & {2.73}
                    & {57.0}
                    & {34.6}
                    & {60.4}
		            & {30.4}
                    & {61.6}
                    & {28.3}
                    \\\midrule 
        \emm \cite{vu2019advent} 
                    & {21.1}
                    & {21.4}
                    & {59.8}
                    & {5.64}
                    & {75.7}
                    & {2.73}
                    & {56.6}
                    & {38.4}
                    & {61.2}
		            & {33.3}
                    & {62.9}
                    & {38.5}
                    \\ 
	    \cct \cite{ouali2020semi} 
                    & {30.9}
                    & {28.2}
                    & {59.1}
                    & {10.1}
                    & {75.9}
                    & {3.60}
                    & {52.4}
                    & {52.3}
                    & {60.6}
		            & {48.7}
                    & {63.8}
                    & {31.2}
                    \\ 
	    \dan \cite{zhang2017deep} 
                    & {34.7}
                    & {25.7}
                    & {56.4}
                    & {15.1}
                    & {76.5}
                    & {3.01}
                    & {57.2}
                    & {27.1}
                    & {62.3}
		            & {25.8}
                    & {63.2}
                    & {30.7}
                    \\
	    \urpc \cite{luo2021efficient} 
                    & {32.2}
                    & {26.9}
                    & {58.9}
                    & {8.14}
                    & {73.2}
                    & {2.68}
                    & {55.5}
                    & {34.6}
                    & {62.4}
		            & {37.8}
                    & {63.0}
                    & {43.1}
                    \\ 
	    \dct \cite{qiao2018deep} 
                    & {36.0}
                    & {24.2}
                    & {58.5}
                    & {10.8}
                    & {78.1}
                    & {2.64}
                    & {57.6}
                    & {38.5}
                    & {60.8}
		            & {34.4}
                    & {61.9}
                    & {31.7}
                    \\ 
	   { \simcvd \cite{you2022simcvd} }
                    & {{32.1}}
                    & {{20.3}}
                    & {{76.1}}
                    & {{4.14}}
                    & {{79.2}}
                    & {{2.21}}
                    & {{56.2}}
                    & {{32.7}}
                    & {{60.5}}
                    & {{23.6}}
                    & {{61.3}}
                    & {{26.0}}
                    \\ 
	    {\mms \cite{lou2023min} }
                    & {{32.5}}
                    & {{13.6}}
                    & {{77.6}}
                    & {{3.61}}
                    & {{79.4}}
                    & {{1.74}}
                    & {{56.9}}
                    & {{45.6}}
                    & {{61.6}}
                    & {{55.4}}
                    & {{62.5}}
                    & {{46.9}}
                    \\ 
	    \ict \cite{verma2019interpolation} 
                    & {35.8}
                    & {21.3}
                    & {59.0}
                    & {4.59}
                    & {75.1}
                    & {0.898}
                    & {58.3}
                    & {32.2}
                    & {60.1}
		            & {39.1}
                    & {62.5}
                    & {32.4}
                    \\ 
	    \mt \cite{tarvainen2017mean} 
                    & {36.8}
                    & {19.6}
                    & {58.3}
                    & {11.2}
                    & {80.1}
                    & {2.33}
                    & {56.7}
                    & {34.3}
                    & {61.9}
		            & {40.0}
                    & {63.3}
                    & {26.2}
                    \\ 
	    \uamt \cite{yu2019uncertainty} 
                    & {35.2}
                    & {24.3}
                    & {61.0}
                    & {7.03}
                    & {77.6}
                    & {3.15}
                    & {57.8}
                    & {41.9}
                    & {61.0}
		            & {47.0}
                    & {62.3}
                    & {26.0}
                    \\ 
	    \cps \cite{chen2021semi} 
                    & {37.1}
                    & {30.0}
                    & {61.0}
                    & {2.92}
                    & {78.8}
                    & {3.41}
                    & {57.7}
                    & {39.6}
                    & {62.1}
		            & {36.0}
                    & {64.0}
                    & {23.6}
                    \\ 
	    \gcl \cite{chaitanya2020contrastive} 
                    & {59.7}
                    & {14.3}
                    & {70.6}
                    & {2.24}
                    & {87.0}
                    & \tf{0.751}
                    & {59.3}
                    & {29.5}
                    & {63.3}
		            & {20.1}
                    & {65.0}
                    & {37.2}
                    \\ 
	    \scs \cite{hu2021semi} 
                    & {59.4}
                    & {12.7}
                    & {73.6}
                    & {5.37}
                    & {84.2}
                    & {2.01}
                    & {57.8}
                    & {39.6}
                    & {61.5}
		            & {28.8}
                    & {64.6}
                    & {33.9}
                    \\ 
	    \plc \cite{chaitanya2021local} 
                    & {58.8}
                    & {15.1}
                    & {70.6}
                    & {2.67}
                    & {87.3}
                    & {1.34}
                    & {56.6}
                    & {41.6}
                    & {62.7}
		            & {26.1}
                    & {68.2}
                    & \tf{16.9}
                    \\ 
        \gr \cb \alg (ours)
                    & \tf{82.6}
                    & \tf{2.03}
                    & \tf{88.8}
                    & \tf{0.622}
                    & \tf{90.7}
                    & {0.864}
                    & \tf{64.1}
                    & \tf{20.9}
                    & \tf{67.3}
		            & \tf{16.4}
                    & \tf{69.3}
                    & {18.0}
                    \\ 
		              \bottomrule
	\end{tabular}
    \end{adjustbox}
    \end{center}
    \vspace{-10pt}
\end{table*}

\begin{figure*}[ht]
\centering
\includegraphics[width=0.9\linewidth]{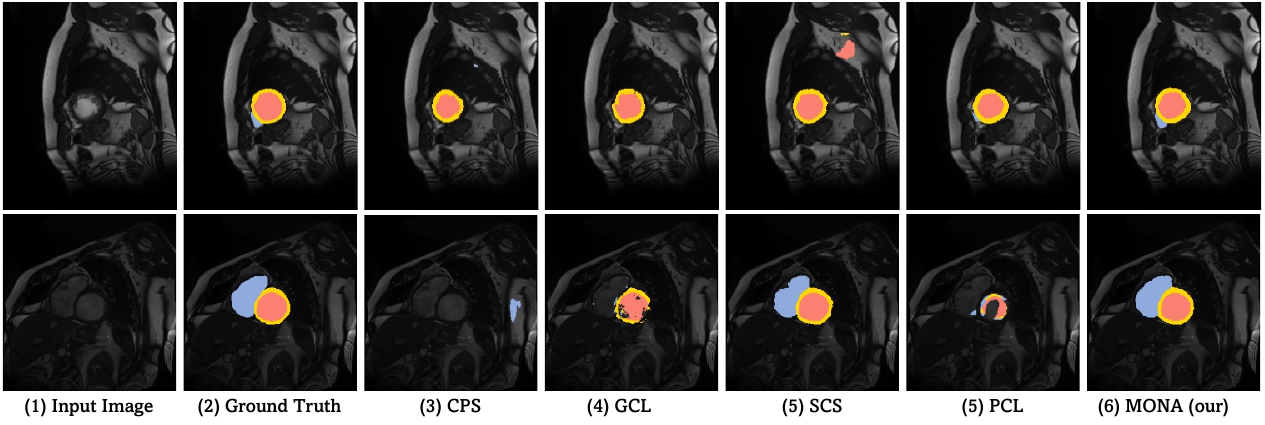}
\vspace{-10pt}
\caption{Visualization of segmentation results on ACDC with 5\% label ratio. As is shown, \alg consistently yields more accurate predictions and better boundary adherence compared to all other SSL methods. Different anatomical classes are shown in different colors (RV: \includegraphics[scale=0.38,valign=c]{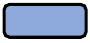}; Myo: \includegraphics[scale=0.4,valign=c]{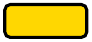}; LV: \includegraphics[scale=0.4,valign=c]{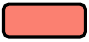}).} 
\label{fig:vis_acdc}
\vspace{-5pt}
\end{figure*}

\begin{figure}[t]
    \centering
    \includegraphics[width=0.7\linewidth]{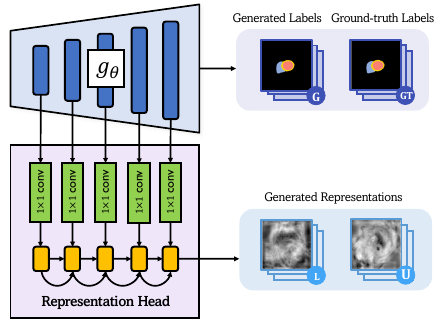}
\vspace{-10pt}
    \caption{Overview of the representation head architecture.}
    \label{fig:rep_head}
    \vspace{-5pt}
\end{figure}

\begin{figure*}[t]
\centering
\includegraphics[width=0.93\linewidth]{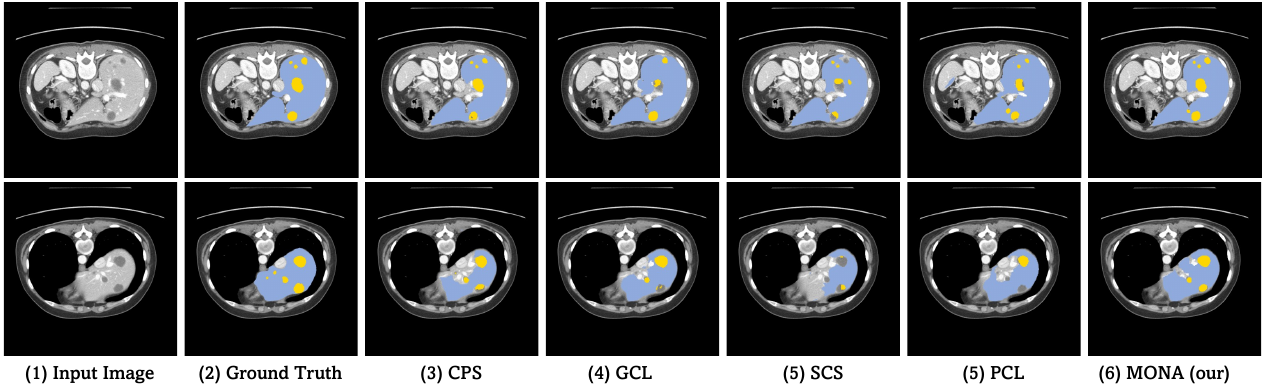}
\vspace{-10pt}
\caption{Visualization of segmentation results on LiTS with 5\% labeled ratio. As is shown, \alg  consistently produces sharp and accurate object boundaries compared to all other SSL methods. Different anatomical classes are shown in different colors (Liver: \includegraphics[scale=0.38,valign=c]{figures/protype_blue.pdf}; Tumor: \includegraphics[scale=0.4,valign=c]{figures/protype_yellow.pdf}).} 
\label{fig:vis_lits}
\vspace{-5pt}
\end{figure*}

In this section, we evaluate our proposed {\alg} on three popular medical image segmentation datasets under varying labeled ratio settings: the ACDC dataset \cite{bernard2018deep}, the LiTS dataset \cite{bilic2019liver}, and the MMWHS dataset \cite{zhuang2016multi}.

\subsection{Datasets}
{\bf The ACDC dataset} was hosted in MICCAI 2017 ACDC challenge \cite{bernard2018deep}, which includes 200 3D cardiac cine MRI scans with expert annotations for three classes (\ie, left ventricle (LV), myocardium (Myo), and right ventricle (RV)). We use 120, 40 and 40 scans for training, validation, and testing\footnote{\url{https://github.com/HiLab-git/SSL4MIS/tree/master/data/ACDC}}. Note that 1\%, 5\%, and 10\% label ratios denote the ratio of patients. For pre-processing, we adopt the similar setting in \cite{chaitanya2020contrastive} by normalizing the intensity of each 3D scan (\ie, using min-max normalization) into $[0,1]$, and re-sampling all 2D scans and the corresponding segmentation maps into a fixed spatial resolution of $256\times 256$ pixels.

{\bf The LiTS dataset} was hosted in MICCAI 2017 Liver Tumor Segmentation Challenge \cite{bilic2019liver}, which includes 131 contrast-enhanced 3D abdominal CT volumes with expert annotations for two classes (\ie, liver and tumor). Note that 1\%, 5\%, and 10\% label ratios denote the ratio of patients. We use 100 and 31 scans for training, and testing with random order. The splitting details are in the supplementary material. For pre-processing, we adopt the similar setting in \cite{li2018h} by truncating the intensity of each 3D scan into $[-200,250]$ HU for removing irrelevant and redundant details, normalizing each 3D scan into $[0,1]$, and re-sampling all 2D scans and the corresponding segmentation maps into a fixed spatial resolution of $256\times 256$ pixels.

{\bf The MMWHS dataset} was hosted in MICCAI 2017 challenge \cite{zhuang2016multi}, which includes 20 3D cardiac MRI scans with expert annotations for seven classes: left ventricle (LV), left atrium (LA), right ventricle (RV), right atrium (RA), myocardium (Myo), ascending aorta (AAo), and pulmonary artery (PA). Note that 1\%, 5\%, and 10\% label ratios denote the ratio of patients. We use 15 and 5 scans for training and testing with random order. The splitting details are in the supplementary material. For pre-processing, we normalize the intensity of each 3D scan (\ie, using min-max normalization) into $[0,1]$, and re-sampling all 2D scans and the corresponding segmentation maps into a fixed spatial resolution of $256\times 256$ pixels.

Moreover, to further validate our approach’s unsupervised imbalance handling ability, we consider a more realistic and more challenging scenario, wherein the models would only have access to the extremely limited labeled data (\ie, 1\% labeled ratio) and large quantities of unlabeled one in training. For all experiments, we follow the same training and testing protocol. Note that 1\%, 5\%, and 10\% label ratios denote the ratio of patients. For ACDC, we adopt the fixed data split \cite{wu2022exploring}. For LiTS and MMWHS, we adopt the random data split with respect to patient.

\subsection{Implementation Details.}
\label{section:implementation}
We implement all the evaluated models using PyTorch library \cite{paszke2019pytorch}. All the models are trained using Stochastic Gradient Descent (SGD) (\ie, initial learning rate = $0.01$, momentum = $0.9$, weight decay = $0.0001$) with batch size of $6$, and the initial learning rate is divided by $10$ every $2500$ iterations. All of our experiments are conducted on NVIDIA GeForce RTX 3090 GPUs. We first train our model with 100 epochs during the pre-training, and then retrain the model for 200 epochs during the fine-tuning. We set the temperature $\tau_\xi$, $\tau_\theta$, $\tau$ as 0.01, 0.1, 0.5. The size of the memory bank is 36. During the pre-training, we follow the settings of \isd, including global projection head setting, and predictors with the $512$-dimensional output embedding, and adopt the setting of local projection head in \cite{hu2021semi}. More specifically, given the predicted logits $\hat{\y} \in \mathbb{R}^{\mathcal{C}\times \mathcal{H}\times \mathcal{W}}$, we create 36 different views (\ie, random crops at the same location) of $\hat{\y}$ and $\hat{\y}'$ with the fixed size $64\times 64$, and then project all pixels into $512$-dimensional output embedding space, and the output feature dimension of $h'_\theta$ is also $512$. An illustration of our representation head is presented in Figure \ref{fig:rep_head}. We then actively sample 256 query embeddings and 512 key embeddings for each mini-batch, and the confidence threshold $\delta_\theta$ is set to $0.97$. When fine-tuning we use an equally sized pool of candidates $K=5$, as well as $\lambda_{1}=0.01$, $\lambda_{2}=1.0$, $\lambda_{3}=1.0$, and $\lambda_{4}=1.0$. For different augmentation strategies, we implement the weak augmentation to the teacher's input as random rotation, random cropping, horizontal flipping, and strong augmentation to the student's input as random rotation, random cropping, horizontal flipping, random contrast, CutMix \cite{french2019semi}, brightness changes \cite{perez2018data}, morphological changes (diffeomorphic deformations). We adopt two popular evaluation metrics: Dice coefficient (DSC) and Average Symmetric Surface Distance (ASD) for 3D segmentation results. Of note, the projection heads, the predictor, and the representation head are only used in training, and will be discarded during inference.

\subsection{Main Results}
\label{section:results}
We show the effectiveness of our method under three different label ratios (\ie, 1\%, 5\%, 10\%). We also compare \alg with various state-of-the-art SSL and fully-supervised methods on three datasets: ACDC \cite{bernard2018deep}, LiTS \cite{bilic2019liver}, MMWHS \cite{zhuang2016multi}. We choose 2D \unet \cite{ronneberger2015u} as backbone, and compare against SSL methods including \emm \cite{vu2019advent}, \cct \cite{ouali2020semi}, \dan \cite{zhang2017deep}, \urpc \cite{luo2021efficient}, \dct \cite{qiao2018deep}, \ict \cite{verma2019interpolation}, \mt \cite{tarvainen2017mean}, \uamt \cite{yu2019uncertainty}, \cps \cite{chen2021semi}, {\simcvd \cite{you2022simcvd}, MMS \cite{lou2023min}}, \scs \cite{hu2021semi}, \gcl \cite{chaitanya2020contrastive}, and \plc \cite{chaitanya2021local}. {The upper bound and lower bound method are \unet trained with full/limited supervisions (\unetF/\unetL), respectively.}
We report quantitative comparisons on ACDC and LiTS in Table \ref{table:acdc_lits_main}.

\begin{table}[ht]
	\begin{center}
	\caption{Comparison of segmentation performance (DSC{[}\%{]}/ASD{[}mm{]}) on MMWHS under three labeled ratio settings (1\%, 5\%, 10\%). On all three labeled settings, \alg significantly outperforms all the state-of-the-art methods by a significant margin. The best results are in \tf{bold}.}
	\label{table:mmwhs_main}
    \begin{adjustbox}{width=0.98\linewidth}
	\begin{tabular}{ccccccc}
		\toprule
		& \multicolumn{2}{c}{1\% Labeled} & \multicolumn{2}{c}{5\% Labeled} & \multicolumn{2}{c}{10\% Labeled} \\
        \cmidrule(r){2-3} \cmidrule(r){4-5} \cmidrule(r){6-7}
		{Method}
		            & {DSC~$\uparrow$}
		            & {ASD~$\downarrow$}
		            & {DSC~$\uparrow$}
		            & {ASD~$\downarrow$}  
		            & {DSC~$\uparrow$}
		            & {ASD~$\downarrow$} 
		            \\ \midrule
		\unetF \cite{ronneberger2015u}
		            & 85.8
                    & 8.01
		            & 85.8
                    & 8.01
		            & 85.8
                    & 8.01
                    \\
		\unetL        
                    & {58.3} 
		            & {33.9}
                    & {77.8}
                    & {24.4} 
		            & 82.7
                    & 13.5
                    \\\midrule 
	    \emm \cite{vu2019advent} 
                    & {54.5}
                    & {41.1}
                    & {80.6}
                    & {17.3}
                    & 82.1
                    & 15.1
                    \\ 
                    
	    \cct \cite{ouali2020semi} 
                    & {62.8}
                    & {27.5}
                    & {79.0}
                    & {21.9}
                    & 79.4
                    & 16.3
                    \\ 
	    \dan \cite{zhang2017deep} 
                    & {52.8}
                    & {48.4}
                    & {79.4}
                    & {22.7}
                    & 80.2
                    & 15.0
                    \\
	    \urpc \cite{luo2021efficient} 
                    & {65.7}
                    & {29.7}
                    & {73.7}
                    & {20.5}
                    & 81.9
                    & 12.3
                    \\ 
	    \dct \cite{qiao2018deep} 
                    & {62.7}
                    & {27.5}
                    & {80.8}
                    & {23.0}
                    & 82.8
                    & 12.4
                    \\ 
	    {\simcvd \cite{you2022simcvd} }
                    & {{64.6}}
                    & {{39.5}}
                    & {{77.0}}
                    & {{20.2}}
                    & {80.3}
                    & {16.8}
                    \\ 
	    {\mms \cite{lou2023min} }
                    & {{66.2}}
                    & {{36.9}}
                    & {{80.6}}
                    & {{18.4}}
                    & {82.1}
                    & {16.7}
                    \\ 
	    \ict \cite{verma2019interpolation} 
                    & {59.9}
                    & {32.8}
                    & {76.5}
                    & {15.4}
                    & 82.2
                    & 12.0
                    \\ 
	    \mt \cite{tarvainen2017mean} 
                    & {58.8}
                    & {35.6}
                    & {76.5}
                    & {15.5}
                    & 79.4
                    & 19.8
                    \\ 
	    \uamt \cite{yu2019uncertainty} 
                    & {61.1}
                    & {37.6}
                    & {76.3}
                    & {20.9}
                    & 83.7
                    & 14.2
                    \\ 
	    \cps \cite{chen2021semi} 
                    & {58.8}
                    & {33.6}
                    & {78.3}
                    & {22.5}
                    & 82.0
                    & 13.1
                    \\ 
	    \gcl \cite{chaitanya2020contrastive} 
                    & {71.6}
                    & {20.3}
                    & {83.5}
                    & \tf{7.41}
                    & 86.7
                    & 8.76
                    \\ 
	    \scs \cite{hu2021semi} 
                    & {71.4}
                    & {19.3}
                    & {81.1}
                    & {11.5}
                    & 82.6
                    & 9.68
                    \\ 
	    \plc \cite{chaitanya2021local} 
                    & {71.5}
                    & {19.8}
                    & {83.4}
                    & {10.7}
                    & {86.0}
                    & {9.65}
                    \\ 
        \gr \cb \alg (ours)
                    & \tf{83.9}
                    & \tf{9.06}
                    & \tf{86.3}
                    & {8.22}
                    & \tf{87.6}
                    & \tf{6.83}
                    \\ 
		              \bottomrule
	\end{tabular}
    \end{adjustbox}
    \end{center}
\end{table}

\begin{figure*}[ht]
\centering
\includegraphics[width=0.93\linewidth]{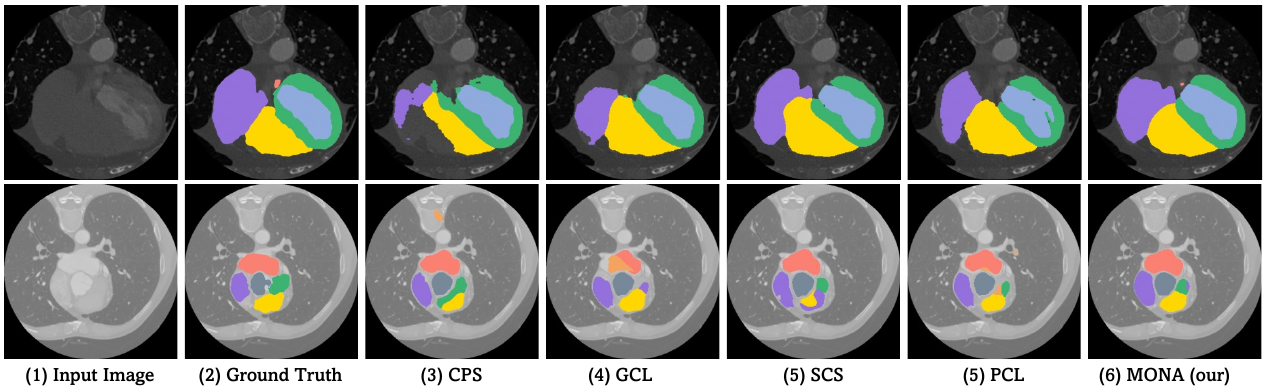}
\vspace{-10pt}
\caption{Visualization of segmentation results on MMWHS with 5\% labeled ratio. As is shown, \alg  consistently generates more accurate predictions compared to all other SSL methods with a significant performance margin. Different anatomical classes are shown in different colors (LV: \includegraphics[scale=0.38,valign=c]{figures/protype_blue.pdf}; LA: \includegraphics[scale=0.4,valign=c]{figures/protype_yellow.pdf}; RV: \includegraphics[scale=0.4,valign=c]{figures/protype_red.pdf}; RA: \includegraphics[scale=0.4,valign=c]{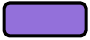}; Myo: \includegraphics[scale=0.4,valign=c]{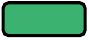}; PA: \includegraphics[scale=0.4,valign=c]{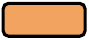}).} 
\label{fig:vis_mm}
\vspace{-5pt}
\end{figure*}

% {\bf ACDC.}
\myparagraph{ACDC.}
We benchmark performances on ACDC with respect to different labeled ratios (\ie, 1\%, 5\%, 10\%). The following observations can be drawn: \underline{\textit{First}}, our proposed \alg significantly outperforms all other SSL methods under three different label ratios. 
Especially, with only extremely limited labeled data available (\eg, 1\%), our method obtains massive gains of {22.9\%} and {10.67} in Dice and ASD (\ie, dramatically improving the performance from 59.7\% to 82.6\%). \underline{\textit{Second}}, as shown in Figure \ref{fig:vis_acdc}, we can see the clear advantage of \alg, where the anatomical boundaries of different tissues are clearly more pronounced such as RV and Myo regions. As seen, our method is capable of producing consistently sharp and accurate object boundaries across various challenge scenarios.

\myparagraph{LiTS.}
We then evaluate \alg on LiTS, using 1\%, 5\%, 10\% labeled ratios. The results are summarized in Table~\ref{table:acdc_lits_main} and Figure \ref{fig:vis_lits}. The conclusions are highly consistent with the above ACDC case: \underline{\textit{First}}, at the different label ratios (\ie, 1\%, 5\%, 10\%), \alg consistently outperforms all the other SSL methods, which again demonstrates the effectiveness of learning representations for the inter-class correlations and intra-class invariances under imbalanced class-distribution scenarios. In particular, our \alg, trained on a 1\% labeled ratio (\ie, extremely limited labels), dramatically improves the previous best averaged Dice score from {59.3\%} to 64.1\% by a large margin, and even performs on par with previous SSL methods using 10\% labeled ratio. \underline{\textit{Second}}, our method consistently outperforms all the evaluated SSL methods under different label ratios (\ie, 1\%, 5\%, 10\%). \underline{\textit{Third}}, as shown in Figure \ref{fig:vis_lits}, we observe that \alg is able to produce more accurate results compared to the previous best schemes.

\myparagraph{MMWHS}
Lastly, we validate \alg on MMWHS, under 1\%, 5\%, 10\% labeled ratios. The results are provided in Table \ref{table:mmwhs_main} and Figure \ref{fig:vis_mm}. 
Again, we found that \alg consistently outperforms all other SSL methods with a significant performance margin, and achieves the highest accuracy among all the SSL approaches under three labeled ratios. 
As is shown, \alg trained at the 1\% labeled ratio significantly outperforms all other methods trained at the 1\% labeled ratio, even over the 5\% labeled ratio. {Concretely, \alg trained at only 1\% labeled ratio outperforms the second-best method (\ie, \gcl) both at the 1\% and 5\% labeled, yielding 12.3\% and {2.8\%} gains in Dice.} We also observe the similar patterns that, \alg performs better or on par with all the other methods at 10\% labeled, which again demonstrates the superiority of \alg in extremely limited labeled data regimes.

Overall, we conclude that \alg provides robust performance on all the medical datasets we evaluated, exceeding that of the fully-supervised baseline, and outperforming all other SSL methods.

\subsection{Ablation Study}
\label{subsection:ablation}
In this subsection, we conduct comprehensive analyses to understand the inner workings of \alg on ACDC under 5\% labeled ratio.

\subsection{Effects of Different Components}
Our key observation is that it is crucial to build meaningful anatomical representations for the inter-class correlations and intra-class invariances under imbalanced class-distribution scenarios can further improve performance. Upon our choice of architecture, we first consider our CL pre-trained method (\ie, \glc). To validate this, we experiment with the key components in \alg on ACDC, including: (1) tailness, (2) consistency, and (3) diversity. The results are in Table \ref{table:component_ablation}. As is shown, each key component makes a clear difference and leveraging all of them contributes to the remarkable performance improvements. This suggests the importance of learning meaningful representations for the inter-class correlations and intra-class invariances within the entire dataset. The intuitions behind each concept are as follows: (1) \textbf{Only \em{tailness}:} many anatomy-rich head classes would be sampled; (2) \textbf{Only \em{consistency}:} it would lead to object collapsing due to the different anatomical variations; (3) \textbf{Only \em{diversity}:} oversampling too many negative samples often comes at the cost of performance degradation. By combining {\em{tailness}}, {\em{consistency}}, and {\em{diversity}}, our method confers a significant advantage at representation learning in imbalanced feature similarity, semantic consistency and anatomical diversity, which further highlights the superiority of our proposed \alg (More results in Section \ref{section:ablation-CL}).

\begin{table}[t]
\caption{Ablation on model component: (1) {{tailness}}; (2) {{consistency}}; (3) {{diversity}}, compared to the Vanilla and our \alg.} 
\label{table:component_ablation}
\vspace{-10pt}
\centering
\resizebox{0.45\textwidth}{!}{
\begin{tabular}{l|c c} 
\toprule
\multirow{2}{*}{Method} & \multicolumn{2}{c}{\textbf{Metrics}} \\ 
    & Dice{[}\%{]}~$\uparrow$
    & ASD{[}mm{]}~$\downarrow$ 
    \\ \midrule
    Vanilla  
    & 74.2
    & 3.89
    \\
    \midrule
    \gr \quad w/ tailness
    & 83.1
    & 0.602
    \\
    \quad w/ consistency
    & {84.2}
    & 1.86
    \\
    \gr \quad w/ diversity
    & {78.2}
    & 3.07 
    \\ 
    \quad w/ tailness + consistency
    & 88.1
    & 0.864
    \\
    \gr \quad w/ consistency + diversity
    & {80.2}
    & 2.11
    \\
    \quad w/ tailness + diversity
    & 85.0
    & 0.913
    \\ 
    \midrule
	\gr \cb \alg (ours)
    & \textbf{88.8}
    & \textbf{0.622}   
    \\ 
    \bottomrule
\end{tabular}
}
\vspace{-10pt}
\end{table}

\subsection{Effects of Different Augmentations}
In addition to further improving the quality and stability in anatomical representation learning, we claim that \alg also gains robustness using augmentation strategies. For augmentation strategies, previous works \cite{tejankar2021isd,zheng2021ressl,sohn2020fixmatch} show that composing the weak augmentation strategy for the ``pivot-to-target'' model (\ie, trained with limited labeled data and a large number of unlabeled data) is helpful for anatomical representation learning since the standard contrastive strategy is too aggressive, intuitively leading to a ``hard'' task (\ie, introducing too many disturbances and yielding model collapses).
Here we examine whether and how applying different data augmentations helps \alg. In this work, we implement the weak augmentation to the teacher's input as random rotation, random cropping, horizontal flipping, and strong augmentation to the student's input as random rotation, random cropping, horizontal flipping, random contrast, CutMix \cite{french2019semi}, brightness changes \cite{perez2018data}, morphological changes (diffeomorphic deformations). We summarize the results in Table \ref{tab:aug_ablation}, and list the following observations: (1) \tf{weak augmentations benefits more}: composing the weak augmentation for the teacher model and strong augmentation for the student model significantly boosts the performance across two benchmark datasets. (2) \tf{same augmentation pairs do not make more gains}: interestingly, applying same type of augmentation pairs does not lead to the best performance compared to different types of augmentation pairs. We postulate that composing different augmentations can be considered as a harder albeit more useful strategy for anatomical representation learning, making feature more generalizable.

\begin{table}[t]
\centering
\caption{Ablation on augmentation strategies for \alg on the ACDC and LiTS dataset under 5\% labeled ratio.
} 
\vspace{-10pt}
\label{tab:aug_ablation}
\resizebox{0.4\textwidth}{!}{
\begin{tabular}{ll|c c|c c} 
\toprule
\multicolumn{2}{c|}{\multirow{2}{*}{Dataset}}
&\multicolumn{2}{c|}{Student Teacher}
&\multicolumn{2}{c}{\textbf{Metrics}} \\ 
& & \multicolumn{2}{c|}{Aug.  Aug.} 
    & Dice{[}\%{]}~$\uparrow$
    & ASD{[}mm{]}~$\downarrow$ 
    \\ \midrule
    \multicolumn{2}{c|}{\multirow{4}{*}{\bf ACDC}}
    & \cgr Weak
    & \cgr Weak
    & \cgr 86.0
    & \cgr 1.02 
    \\ & 
    & Strong 
    & Weak
    & {88.8}
    & {0.622}
    \\ &
    & \cgr Weak
    & \cgr Strong
    & \cgr 86.4
    & \cgr 2.83
    \\  &
    & Strong
    & Strong
    & 88.8
    & 2.07
    \\ 
    \midrule
    \multicolumn{2}{c|}{\multirow{4}{*}{\bf LiTS}}
    & \cgr Weak
    & \cgr Weak
    & \cgr 62.3
    & \cgr 26.5
    \\ &
    & {Strong}
    & {Weak}
    & {67.3}
    & {16.4}
    \\ &
    & \cgr {Weak}
    & \cgr {Strong}
    & \cgr 64.3
    & \cgr 34.7
    \\  &
    & Strong
    & Strong
    & 66.5
    & 21.1
    \\ 
    \bottomrule
\end{tabular}
}
\vspace{-10pt}
\end{table}

\begin{table*}[t]
	\begin{center}
	\caption{Ablation study of different contrastive learning frameworks on ACDC under three labeled ratio settings (1\%, 5\%, 10\%). We compare two settings: with or without {\em{fine-tuning}} on the segmentation performance (DSC{[}\%{]}/ASD{[}mm{]}). We denote `without {\em{fine-tuning}}'' to only {\em{pretaining}}. On all three labeled settings, our methods (\ie, \glc and \alg) significantly outperform all the state-of-the-art methods by a significant margin. All the experiments are run with three different random seeds. The best results are in \tf{bold}.}
	\label{table:cl_comparsion}
    \begin{adjustbox}{width=0.7\linewidth}
	\begin{tabular}{cccccccc}
		\toprule
		& & \multicolumn{2}{c}{1\% Labeled} & \multicolumn{2}{c}{5\% Labeled} & \multicolumn{2}{c}{10\% Labeled} \\
        \cmidrule(r){3-4} \cmidrule(r){5-6} \cmidrule(r){7-8}
        % \vspace{0.15in}
		Framework & {Method}
		            & {DSC~$\uparrow$}
		            & {ASD~$\downarrow$}
		            & {DSC~$\uparrow$}
		            & {ASD~$\downarrow$}  
		            & {DSC~$\uparrow$}
		            & {ASD~$\downarrow$} 
		            \\ \midrule
		\multirow{6}{*}{only {\em{pretaining}}}
		&\mocoo\cite{chen2020improved}
                    & {38.6} 
		            & {22.4}
                    & {56.2}
                    & {17.9} 
		            & {81.0}
                    & {5.36}
                    \\
		&\mocoknn \cite{van2021revisiting}
		            & 39.5
                    & 22.0
		            & 58.3
                    & 15.7
		            & {83.1}
                    & {7.18}
                    \\
		&\simclr\cite{chen2020simple}
		            & 34.8
                    & 24.3
		            & 51.7
                    & 19.9
		            & 80.3
                    & 4.16
                    \\                    
		&\byol\cite{grill2020bootstrap}
		            & 35.9
                    & {7.25}
		            & 65.9
                    & 9.15
		            & 85.6
                    & 2.51
                    \\                       
		&\isd\cite{tejankar2021isd}
		            & {45.8}
                    & {17.2}
		            & {71.0}
                    & {4.29}
		            & {85.3}
                    & {2.97}
                    \\                    
		& \cgr \vb {\glc (ours)}
		            & { \cgr \tf{49.3}}
                    & { \cgr \tf{7.11}}
		            & { \cgr \tf{74.2}}
                    & { \cgr \tf{3.89}}
		            & { \cgr \tf{86.5}}
                    & { \cgr \tf{1.92}}
                    \\                    
                    \midrule
		\multirow{6}{*}{w/ \textit{fine-tuning}}
		&\mocoo\cite{chen2020improved}
                    & {77.7}
                    & {4.78}
                    & {85.4}
                    & {1.52}
                    & {86.7}
                    & {1.74}
                    \\ 
		&\mocoknn \cite{van2021revisiting}
                    & {78.0}
                    & {4.28}
                    & {85.9}
                    & {1.51}
                    & 86.9
                    & 1.61
                    \\ 
		&\simclr\cite{chen2020simple}
                    & {75.7}
                    & {4.33}
                    & {83.2}
                    & {2.06}
                    & 86.1
                    & 2.25
                    \\
		&\byol\cite{grill2020bootstrap}
                    & {77.1}
                    & {4.84}
                    & {85.3}
                    & {2.06}
                    & 88.1
                    & 0.994
                    \\ 
		&\isd\cite{tejankar2021isd}
                    & {80.1}
                    & {3.00}
                    & {83.8}
                    & {1.95}
                    & {88.6}
                    & {1.20}
                    \\ 
        &\cgr \cb \alg (ours)
                    & \cgr \tf{82.6}
                    & \cgr \tf{2.03}
                    & \cgr \tf{88.8}
                    & \cgr \tf{0.622}
                    & \cgr \tf{90.7}
                    & \cgr \tf{0.864}
                    \\ 
		              \bottomrule
	\end{tabular}
    \end{adjustbox}
    \end{center}
    \vspace{-10pt}
\end{table*}

\begin{table*}[t]
	\begin{center}
	\caption{Ablation study of different principles across different contrastive learning frameworks under various labeled ratio settings (1\%, 5\%, 10\%). Experiments are conducted on ACDC using \unet \cite{ronneberger2015u} as the backbone with three independent runs. Here we report the segmentation performance in terms of DSC{[}\%{]} and ASD{[}mm{]}. On all three labeled settings, incorporating our methods (\ie, tailness, consistency, and diversity) consistently achieve superior model robustness gains across different state-of-the-art CL frameworks.} %The best results are highlighted in \tf{bold}.}
	\vspace{-5pt}
	\label{table:cl_principle_comparsion}
    \begin{adjustbox}{width=0.8\linewidth}
	\begin{tabular}{cccccccc}
		\toprule
		& & \multicolumn{2}{c}{1\% Labeled} & \multicolumn{2}{c}{5\% Labeled} & \multicolumn{2}{c}{10\% Labeled} \\
        \cmidrule(r){3-4} \cmidrule(r){5-6} \cmidrule(r){7-8}
        % \vspace{0.15in}
		Framework & {Principle}
		            & {DSC~$\uparrow$}
		            & {ASD~$\downarrow$}
		            & {DSC~$\uparrow$}
		            & {ASD~$\downarrow$}  
		            & {DSC~$\uparrow$}
		            & {ASD~$\downarrow$} 
		            \\ \midrule
		\multirow{8}{*}{\mocoo\cite{chen2020improved}}
		&  Vanilla
                    & {38.6}
                    & {22.4}
                    & {56.2}
                    & {17.9}
                    & {81.0}
                    & {5.36}
                    \\
		& tailness
                    & {65.0}
                    & {3.99}
                    & {81.3}
                    & {1.13}
                    & {84.8}
                    & {1.52}
                    \\
		& consistency
                    & {70.3}
                    & {6.88}
                    & {79.5}
                    & {3.65}
                    & {81.9}
                    & {3.79}
                    \\ 
		& diversity
                    & {47.5}
                    & {10.2}
                    & {72.2}
                    & {5.82}
                    & {83.1}
                    & {5.46}
                    \\ 
		& tailness + consistency
                    & {75.8}
                    & {5.10}
                    & {83.8}
                    & {1.89}
                    & {85.7}
                    & {2.81}
                    \\ 
		& consistency + diversity
                    & {73.3}
                    & {6.34}
                    & {75.4}
                    & {5.63}
                    & {82.7}
                    & {4.39}
                    \\ 
		& tailness + diversity
                    & {75.5}
                    & {5.40}
                    & {82.4}
                    & {3.39}
                    & {85.3}
                    & {2.49}
                    \\ 
		& \cgr tailness + consistency + diversity
                    & \cgr {77.7}
                    & \cgr {4.78}
                    & \cgr {85.4}
                    & \cgr {1.52}
                    & \cgr {86.7}
                    & \cgr {1.74}
                    \\ 
                    \midrule
		\multirow{8}{*}{\mocoknn \cite{van2021revisiting}}
		&  Vanilla
                    & {39.5}
                    & {22.0}
                    & {58.3}
                    & {15.7}
                    & {83.1}
                    & {7.18}
                    \\
		& tailness
                    & {66.7}
                    & {3.87}
                    & {83.7}
                    & {1.39}
                    & {86.2}
                    & {1.17}
                    \\
		& consistency
                    & {72.2}
                    & {5.97}
                    & {81.7}
                    & {3.13}
                    & {84.8}
                    & {3.57}
                    \\ 
		& diversity
                    & {50.5}
                    & {9.53}
                    & {73.5}
                    & {5.92}
                    & {83.5}
                    & {5.45}
                    \\ 
		& tailness + consistency
                    & {76.3}
                    & {4.51}
                    & {84.3}
                    & {2.51}
                    & {85.7}
                    & {2.72}
                    \\ 
		& consistency + diversity
                    & {72.1}
                    & {6.45}
                    & {78.6}
                    & {5.56}
                    & {84.6}
                    & {4.08}
                    \\ 
		& tailness + diversity
                    & {75.5}
                    & {5.75}
                    & {81.7}
                    & {3.01}
                    & {85.6}
                    & {2.14}
                    \\ 
		& \cgr tailness + consistency + diversity
                   & \cgr {78.0}
                    & \cgr {4.28}
                    & \cgr {85.9}
                    & \cgr {1.51}
                    & \cgr {86.9}
                    & \cgr {1.61}
                    \\ 
                    \midrule
		\multirow{8}{*}{\simclr\cite{chen2020simple}}
		&  Vanilla
                    & {34.8}
                    & {24.3}
                    & {51.7}
                    & {19.9}
                    & {80.3}
                    & {4.16}
                    \\
		& tailness
                    & {61.9}
                    & {3.52}
                    & {79.8}
                    & {1.70}
                    & {84.5}
                    & {2.01}
                    \\
		& consistency
                    & {70.8}
                    & {5.46}
                    & {78.1}
                    & {2.89}
                    & {84.7}
                    & {2.24}
                    \\ 
		& diversity
                    & {45.9}
                    & {8.49}
                    & {68.3}
                    & {6.46}
                    & {83.5}
                    & {3.92}
                    \\ 
		& tailness + consistency
                    & {73.0}
                    & {4.24}
                    & {83.0}
                    & {2.43}
                    & {85.9}
                    & {2.46}
                    \\ 
		& consistency + diversity
                    & {71.1}
                    & {6.49}
                    & {75.6}
                    & {4.47}
                    & {83.9}
                    & {3.51}
                    \\ 
		& tailness + diversity
                    & {71.9}
                    & {4.98}
                    & {81.1}
                    & {2.92}
                    & {85.3}
                    & {2.94}
                    \\ 
		& \cgr tailness + consistency + diversity
                    & \cgr {75.7}
                    & \cgr {4.33}
                    & \cgr {83.2}
                    & \cgr {2.06}
                    & \cgr {86.1}
                    & \cgr {2.25}
                    \\ 
                    \midrule
		\multirow{8}{*}{\byol\cite{grill2020bootstrap}}
		&  Vanilla
                    & {35.9}
                    & {7.25}
                    & {65.9}
                    & {9.15}
                    & {85.6}
                    & {2.51}
                    \\
		& tailness
                    & {64.2}
                    & {4.26}
                    & {81.9}
                    & {1.71}
                    & {86.4}
                    & {0.871}
                    \\
		& consistency
                    & {71.0}
                    & {5.45}
                    & {80.2}
                    & {3.22}
                    & {87.0}
                    & {2.08}
                    \\ 
		& diversity
                    & {47.5}
                    & {6.29}
                    & {70.7}
                    & {5.48}
                    & {85.7}
                    & {2.36}
                    \\ 
		& tailness + consistency
                    & {73.7}
                    & {4.74}
                    & {83.3}
                    & {2.01}
                    & {87.7}
                    & {1.25}
                    \\ 
		& consistency + diversity
                    & {70.9}
                    & {6.08}
                    & {76.0}
                    & {4.55}
                    & {86.1}
                    & {1.93}
                    \\ 
		& tailness + diversity
                    & {72.2}
                    & {5.81}
                    & {82.6}
                    & {3.12}
                    & {86.4}
                    & {1.33}
                    \\ 
		& \cgr tailness + consistency + diversity
                    & \cgr {77.1}
                    & \cgr {4.84}
                    & \cgr {85.3}
                    & \cgr {2.06}
                    & \cgr {88.1}
                    & \cgr {0.994}
                    \\ 
                    \midrule
		\multirow{8}{*}{\isd\cite{tejankar2021isd}}
		&  Vanilla
                    & {45.8}
                    & {17.2}
                    & {71.0}
                    & {4.29}
                    & {85.3}
                    & {2.97}
                    \\
		& tailness
                    & {71.8}
                    & {2.80}
                    & {79.2}
                    & {1.47}
                    & {87.1}
                    & {1.02}
                    \\
		& consistency
                    & {78.8}
                    & {3.98}
                    & {80.2}
                    & {2.90}
                    & {87.3}
                    & {1.94}
                    \\ 
		& diversity
                    & {54.5}
                    & {8.03}
                    & {77.1}
                    & {6.90}
                    & {86.2}
                    & {2.58}
                    \\ 
		& tailness + consistency
                    & {79.6}
                    & {2.99}
                    & {83.0}
                    & {1.93}
                    & {88.2}
                    & {1.24}
                    \\ 
		& consistency + diversity
                    & {75.1}
                    & {4.72}
                    & {77.8}
                    & {3.65}
                    & {86.5}
                    & {2.45}
                    \\ 
		& tailness + diversity
                    & {74.8}
                    & {7.98}
                    & {82.3}
                    & {2.02}
                    & {87.2}
                    & {1.35}
                    \\ 
		& \cgr tailness + consistency + diversity
                    & \cgr {80.1}
                    & \cgr {3.00}
                    & \cgr {83.8}
                    & \cgr {1.95}
                    & \cgr {88.6}
                    & \cgr {1.20}
                    \\ 
                    \midrule
		\multirow{8}{*}{\alg (ours)}
		&  Vanilla
                    & {49.3}
                    & {7.11}
                    & {74.2}
                    & {3.89}
                    & {86.5}
                    & {1.92}
                    \\
		& tailness
                    & {75.1}
                    & {1.83}
                    & {83.1}
                    & {0.602}
                    & {87.8}
                    & {0.577}
                    \\
		& consistency
                    & {81.5}
                    & {2.78}
                    & {84.2}
                    & {1.86}
                    & {88.4}
                    & {1.33}
                    \\ 
		& diversity
                    & {62.8}
                    & {3.97}
                    & {78.2}
                    & {3.07}
                    & {86.6}
                    & {1.88}
                    \\ 
		& tailness + consistency
                    & {81.2}
                    & {2.19}
                    & {88.1}
                    & {0.864}
                    & {90.1}
                    & {0.966}
                    \\ 
		& consistency + diversity
                    & {81.8}
                    & {3.29}
                    & {80.2}
                    & {2.11}
                    & {86.9}
                    & {1.67}
                    \\ 
		& tailness + diversity
                    & {78.6}
                    & {3.33}
                    & {85.0}
                    & {0.913}
                    & {89.5}
                    & {0.673}
                    \\ 
		& \cgr tailness + consistency + diversity
                    & \cgr {82.6}
                    & \cgr {2.03}
                    & \cgr {88.8}
                    & \cgr {0.622}
                    & \cgr {90.7}
                    & \cgr {0.864}
                    \\ 
		              \bottomrule
	\end{tabular}
    \end{adjustbox}
    \end{center}
    \vspace{-10pt}
\end{table*}

\subsection{Generalization across Contrastive Learning Frameworks}
\label{section:ablation-CL}
As discussed in Section \ref{subsection:framework}, our motivation comes from the observation that there are only very limited labeled data and a large amount of unlabeled data in real-world clinical practice. As the fully-supervised methods generally outperform all other SSL methods by clear margins, we postulate that leveraging massive unlabeled data usually introduces additional noise during training, leading to degraded segmentation quality. To address this challenge, ``contrastive learning'' is a straightforward way to leverage existing unlabeled data in the learning procedure.
As supported in Section \ref{section:experiments}, our findings have shown that \alg generalizes well across different benchmark datasets (\ie, ACDC, LiTS, MMWHS) with diverse labeled settings (\ie, 1\%, 5\%, 10\%). In the following subsection, we further demonstrate that our proposed principles (\ie, tailness, consistency, diversity) are beneficial to various state-of-the-art CL-based frameworks (\ie, \mocoo\cite{chen2020improved}, \mocoknn \cite{van2021revisiting}, \simclr\cite{chen2020simple}, \byol\cite{grill2020bootstrap}, and \isd\cite{tejankar2021isd}) with different label settings. More details about these three principles can be found in Section \ref{subsection:acr}. 
{Of note, \alg can consistently outperform the semi-supervised methods on diverse benchmark datasets with only 10\% labeled ratio.}

\begin{figure}[h]
\begin{center}
\centerline{
\includegraphics[width=0.85\linewidth]{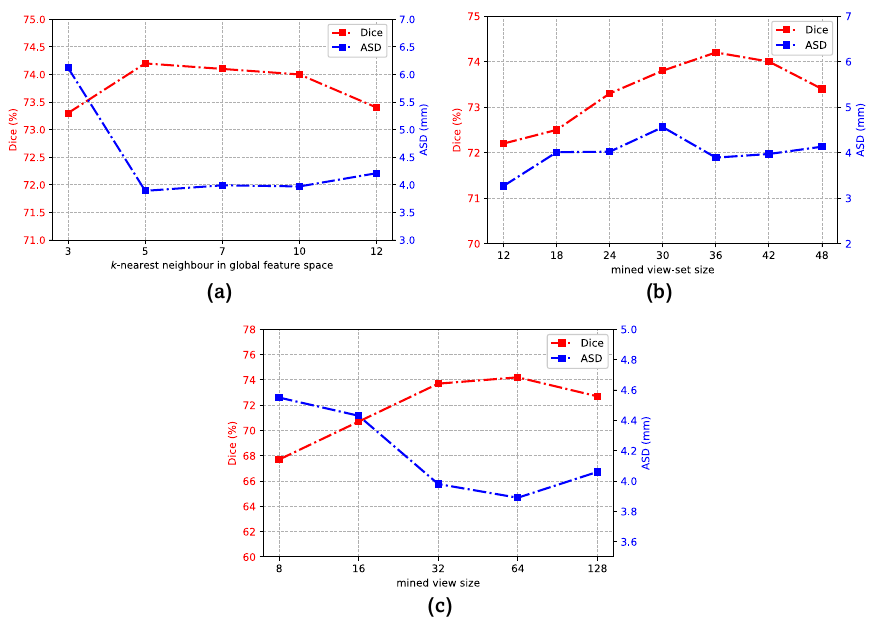}
}
\vspace{-10pt}
\caption{Effects of $k$-nearest neighbour in global feature space,  mined view-set size, and mined view patch size. We report Dice and ASD of \glc on the ACDC dataset at the 5\% labeled ratio. All the experiments are run with three different random seeds.}
\label{fig:ablation-glc}
\end{center}
\end{figure}

\myparagraph{Training Details of Competing CL Methods.}
We identically follow the default setting in each CL framework \cite{chen2020improved,van2021revisiting,chen2020simple,grill2020bootstrap,tejankar2021isd} except the epochs number. We train each model in the semi-supervised setting. For labeled data, we follow the same training strategy in Section \ref{subsection:framework}. As for unlabeled data, we strictly follow the default settings in each baseline. Specifically, for fair comparisons, we pre-train each CL baseline and our CL pre-trained method (\ie, \glc) for 100 epochs in all our experiments. Then we fine-tune each CL model with our proposed principles with the same setting, as provided in Section \ref{section:implementation}. For \mocoknn \cite{van2021revisiting}, given the following ablation study we set the number of neighbors $k$ as 5, and further compare different settings of $k$ in \mocoknn \cite{van2021revisiting} in the following subsection. All the experiments are run with three different random seeds, and the results we present are calculated from the validation set. {Of note, \unetF is fully supervised.}

\myparagraph{Comparisons with CL-based Frameworks.}
Table~\ref{table:cl_comparsion} presents the comparisons between our methods (\ie, \glc and \alg) and various CL baselines.  
After analyzing these extensive results, we can draw several consistent observations. 
\underline{\textit{First}}, we can observe that our \glc achieves performance gains under all the labeled ratios, which not only demonstrates the effectiveness of our method, but also further verifies this argument using ``global-local'' strategy \cite{chaitanya2020contrastive}. 
The average improvement in Dice obtained by \glc could reach up to 2.53\%, compared to the second best scores at different labeled ratios. 
\underline{\textit{Second}}, we can find that incorporating our proposed three principles significantly outperforms the CL baselines without fine-tuning, across all frameworks and different labeled ratios. 
These experimental findings suggest that our proposed three principles can further improve the generalization across different labeled ratios. 
On the ACDC dataset at the 1\% labeled ratio, the backbones equipped with all three principles all obtain promising results, improving the performance of \mocoo, \mocoknn, \simclr, \byol, \isd, and our \glc by 39.1\%,  38.5\%, 40.9\%, 41.2\%, 34.3\%, 33.3\%, respectively. 
The ACDC dataset is a popular multi-class medical image segmentation dataset, with massive imbalanced or long-tailed class distribution cases. 
The imbalanced or long-tailed class distribution gap could result in the vanilla models overfitting to the head class, and generalizing very poorly to the tail class. 
With the addition of under-sampling the head classes, the principle -- \textit{tailness} -- can be deemed as the prominent strategy to yield better generalization and segmentation performance of the models across different labeled ratios. 
Similar results are found under 5\% and 10\% labeled ratios. \underline{\textit{Third}}, over a wide range of labeled ratios, \alg can establish the new state-of-the-art performance bar for semi-supervised 2D medical image segmentation. 
{Particularly, \alg – for the first time – boosts the segmentation performance with 10\% labeled ratio over the fully-supervised \unet (\unetF). From Table~\ref{table:acdc_lits_main} we see that on LiTS with 10\% labeled ratio, \alg outperforms \unetF by 0.8 in terms of DSC (69.3 vs 68.5). From table \ref{table:mmwhs_main}, \alg outperforms \unetF on MMWHS by 1.8 in terms of DSC (87.6 vs 85.8).}
Table~\ref{table:acdc_lits_main} and \ref{table:mmwhs_main} also show that \alg significantly outperforms all the other semi-supervised methods by a large margin.
In summary, our methods (\ie, {\glc} and {\alg}) obtain remarkable performance on all labeled settings. The results verify the superiority of our proposed three principles (\ie, tailness, consistency, diversity) jointly, which makes the model well generalize to different labeled settings, and can be easily and seamlessly plugged into all other CL frameworks  \cite{chen2020improved,van2021revisiting,chen2020simple,grill2020bootstrap,tejankar2021isd} adopting the two-branch design, demonstrating that these concepts consistently help the model yield extra performance boosts for them all.

\begin{figure}[h]
\begin{center}
\centerline{
\includegraphics[width=0.85\linewidth]{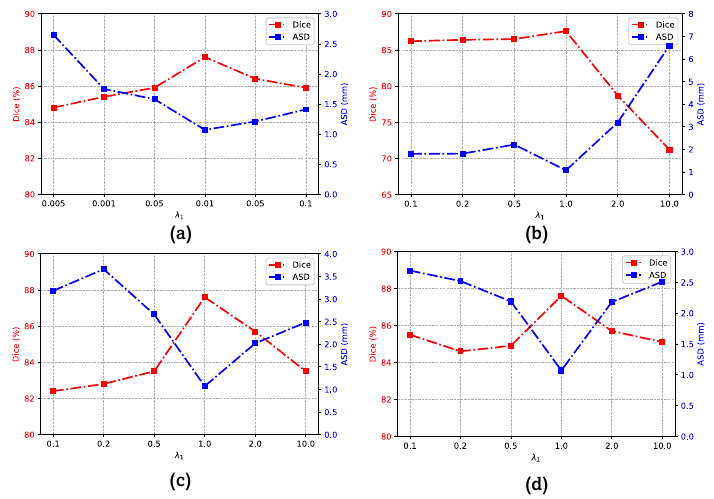}
}
\vspace{-10pt}
\caption{Effects of hyperparameters $\lambda_{1}, \lambda_{2}, \lambda_{3}, \lambda_{4}$. We report Dice and ASD of \alg on the ACDC dataset at the 5\% labeled ratio. All the experiments are run with three different random seeds.}
\label{fig:hyper}
\end{center}
\end{figure}

\myparagraph{Generalization Across CL Frameworks.}
As demonstrated in Table \ref{table:cl_principle_comparsion}, incorporating {\em{tailness}}, {\em{consistency}}, and {\em{diversity}} have obviously superior performance boosts, which is aligned with consistent observations with Section \ref{subsection:ablation} can be drawn. This suggests that these three principles can serve as desirable properties for medical image segmentation in both supervised and unsupervised settings.

\myparagraph{Does $k$-nearest neighbour in global feature space help?}
Prior work suggests that the use of stronger augmentations and nearest neighbour can be the very effective tools in learning additional invariances \cite{van2021revisiting}. That is, both the specific number of nearest neighbours and specific augmentation strategies are necessary to achieve superior performance. In this subsection, we study the relationship of $k$-nearest neighbour in global feature space and the behavior of our \glc for the downstream medical image segmentation. Here we first follow the same augmentation strategies in \cite{van2021revisiting} (More analysis on data augmentation can be found in Section \ref{subsection:ablation}), and then conduct ablation studies on how the choices of $k$-nearest neighbour can influence the performance of \glc. Specifically, we run \glc on the ACDC dataset at the 5\% labeled ratio with a range of $k \in \{3, 5, 7, 10, 12\}$. Figure \ref{fig:ablation-glc}(a) shows the ablation study on $k$-nearest neighbour in global feature on the segmentation performance. As is shown, we find that {\glc} at $k = 5, 7, 10$ have almost identical performance ($k=5$ has slightly better performance compared to other two settings), and all have superior performance compared to all others. In contrast, {\glc} -- through the use of randomly selected samples -- is capable of finding diverse yet semantically consistent anatomical features from the entire dataset, which at the same time gives better segmentation performance.

\myparagraph{Ablation Study of Mined View-Set Size.}
We then conduct ablation studies on how the mined view-set size in \glc can influence the segmentation performance. We run \glc on the ACDC dataset at 5\% labeled ratio with a range of the mined view-set size $\in \{12, 18, 24, 30, 36, 42, 48\}$. The results are summarized in Figure \ref{fig:ablation-glc}(b). As is shown, we find that \glc trained with view-set size $36$ and $42$ have similar or superior performance compared to all other settings, and our model with view-set size of $36$ achieves the highest performance.

\myparagraph{Ablation Study of Mined View Size.}
Lastly, we study the influence of mined view size on the segmentation performance. Specifically, we run \glc on the ACDC dataset at the 5\% labeled ratio with a range of the mined view size $\in \{8, 16, 32, 64, 128\}$. Figure \ref{fig:ablation-glc}(c) shows the ablation study of mined view size on the segmentation performance. As is shown, we observe that {\glc} trained with mined view size of $32$ and $64$ have similar segmentation abilities, and both achieve superior performance compared to other settings. Here the mined view size of $64$ works the best for \glc to yield the superior segmentation performance.

\myparagraph{Conclusion.}
Given the above ablation study, we set $k$, mined view-set size, patch size as $5$, $36$, $64\times 64$ in our experiments, respectively. This can contribute to satisfactory segmentation performance.

\subsection{Ablation Study of Anatomical Contrastive Reconstruction}
\label{section:ablation-acr}
In this section, we give a detailed analysis on the choice of the parameters in the anatomical contrastive reconstruction fine-tuning, and take a deeper look and understand how they contribute to the final segmentation performance. All the hyperparameters in training are the same across three benchmark datasets. All the experiments are run with three different random seeds, and the experimental results we report are calculated from the validation set.

\myparagraph{Ablation Study of Total Loss $\mathcal{L}_\text{total}$.}
Proper choices of hyperparameters in total loss $\mathcal{L}_\text{total}$ (See Section \ref{subsection:acr}) play a significant role in improving overall segmentation quality. We hence conduct the fine-grained analysis of the hyperparameters in $\mathcal{L}_\text{total}$. In practice, we fine-tune the models with three independent runs, and grid search to select multiple hyperparameters. Specifically, we run \alg on the ACDC dataset at the 5\% labeled ratio with a range of different hyperparameters $\lambda_{1}\!\in\!\{0.005, 0.001, 0.05, 0.01, 0.05, 0.1\}$, and $\lambda_{2}, \lambda_{3}, \lambda_{4}\!\in\!\{0.1, 0.2, 0.5, 1.0, 2.0, 10.0\}$. We summarize the results in Figure \ref{fig:hyper}, and take the best setting $\lambda_{1}\!=\!0.01$, $\lambda_{2}\!=\!1.0$, $\lambda_{3}\!=\!1.0, \lambda_{4}\!=\!1.0$.

\myparagraph{Ablation Study of Confidence Threshold ${\delta_\theta}$.}
We then assess the influence of ${\delta_\theta}$ on the segmentation performance. Specifically, we run \alg on the ACDC dataset at the 5\% labeled ratio with a range of the confidence threshold ${\delta_\theta} \in \{0.85, 0.88, 0.91, 0.94, 0.97, 1.0\}$. Figure \ref{fig:ablation}(a) shows the ablation study of ${\delta_\theta}$ on the segmentation performance. As we can see, \alg on ${\delta_\theta}=0.97$ has superior performance compared to other settings.

\begin{figure}[h]
\begin{center}
\centerline{
\includegraphics[width=0.85\linewidth]{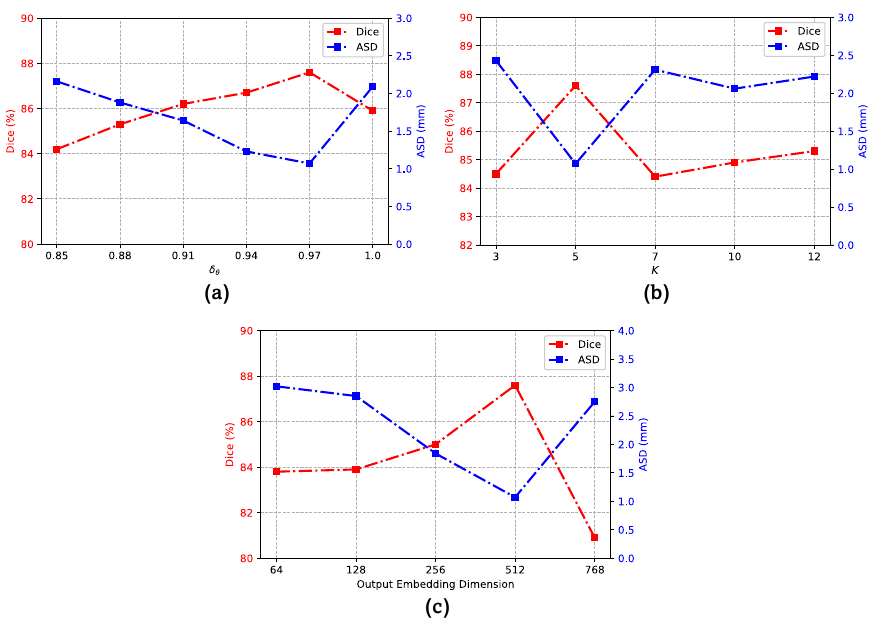}
}
\vspace{-10pt}
\caption{Effects of confidence threshold ${\delta_\theta}$, $K$-nearest neighbour constraint, and output embedding dimension. We report Dice and ASD of \alg on the ACDC dataset at the 5\% labeled ratio. All the experiments are run with three different random seeds.}
\label{fig:ablation}
\end{center}
\end{figure}

\myparagraph{Ablation Study of $K$-Nearest Neighbour Constraint.}
Next, we conduct ablation studies on how the choices of $K$ in $K$-nearest neighbour constraint can influence the segmentation performance. Specifically, we run \alg on the ACDC dataset at the 5\% labeled ratio with a range of the choices $K \in \{3, 5, 7, 10, 12\}$. Figure \ref{fig:ablation}(b) shows the ablation study of $K$ choices on the segmentation performance. As we can see, \alg on $K=5$ achieves the best performance compared to other settings.

\myparagraph{Ablation Study of Output Embedding Dimension.}
Finally, we study the influence of the output embedding dimension on the segmentation performance of \alg. Specifically, we run \alg on the ACDC dataset at the 5\% labeled ratio with a range of output embedding dimension $\in \{64, 128, 256, 512, 768\}$. Figure \ref{fig:ablation}(c) shows the ablation study of output embedding dimension on the segmentation performance. As we can see, \alg with output embedding dimension of $512$, can be trained to outperform other settings.

\myparagraph{Conclusion.}
Given the above ablation study, we select $\lambda_{1}\!=\!0.01$, $\lambda_{2}\!=\!1.0$, $\lambda_{3}\!=\!1.0, \lambda_{4}\!=\!1.0$, ${\delta_\theta}=0.97$, $K=5$, output embedding dimension $=512$ in our experiments. This can provide the optimal segmentation performance across different labeled ratios.
\section{Conclusion}
In this paper, we have presented \alg, a semi-supervised contrastive learning method for 2D medical image segmentation. We start from the observations that medical image data always exhibit a long-tail class distribution, and the same anatomical objects (\ie, liver regions for two people) are more similar to each other than different objects (\eg liver and tumor regions). We further expand upon this idea by introducing anatomical contrastive formulation, as well as equivariance and invariances constraints. Both empirical and theorical studies show that we can formulate a generic set of perspectives that allows us to learn meaningful representations across different anatomical features, which can dramatically improve the segmentation quality and alleviate the training memory bottleneck. 
Extensive experiments on three datasets demonstrate the superiority of our proposed framework in the long-tailed medical data regimes with extremely limited labels. We believe our results contribute to a better understanding of medical image segmentation and point to new avenues for long-tailed medical image data in realistic clinical applications.

\ifCLASSOPTIONcaptionsoff
  \newpage
\fi

% {\small
\bibliographystyle{IEEEtran}
\bibliography{tpami}
% }

\begin{IEEEbiography}[{\includegraphics[width=1in,height=1.25in,clip,keepaspectratio]{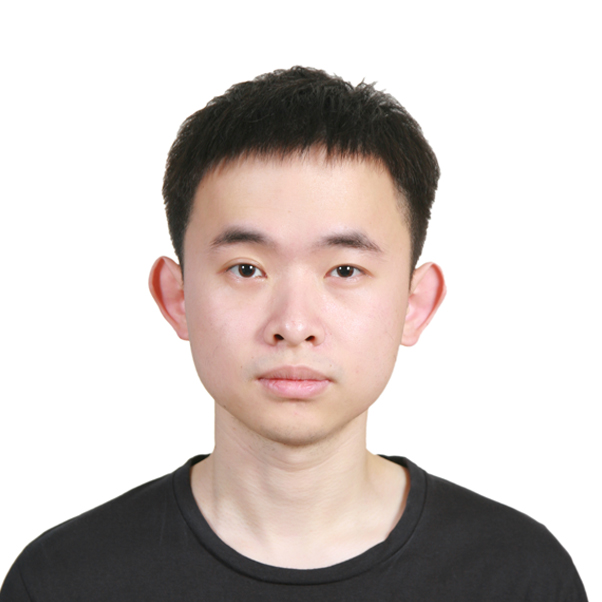}}]{Chenyu You} is a Ph.D. student at Yale University. He received the B.S. and M.Sc. degree from Rensselaer Polytechnic Institute and Stanford University, respectively. His research interests are broadly in the area of machine learning theory and algorithms intersecting the fields of computer \& medical vision, natural language processing, and signal processing. Previously, he has published several works in NeurIPS, ICLR, CVPR, MICCAI, ACL, EMNLP, NAACL, AAAI, IJCAI, and KDD, etc. He has been awarded as the Outstanding/Distinguished Reviewer of CVPR, MICCAI, IEEE Transactions on Medical Imaging (TMI), and Medical Physics. He serves as an area chair at MICCAI, and has co-led workshops such as foundation model on healthcare (ICML).
\end{IEEEbiography}
\vspace{-20pt}

\begin{IEEEbiography}[{\includegraphics[width=1in,height=1.25in,clip,keepaspectratio]{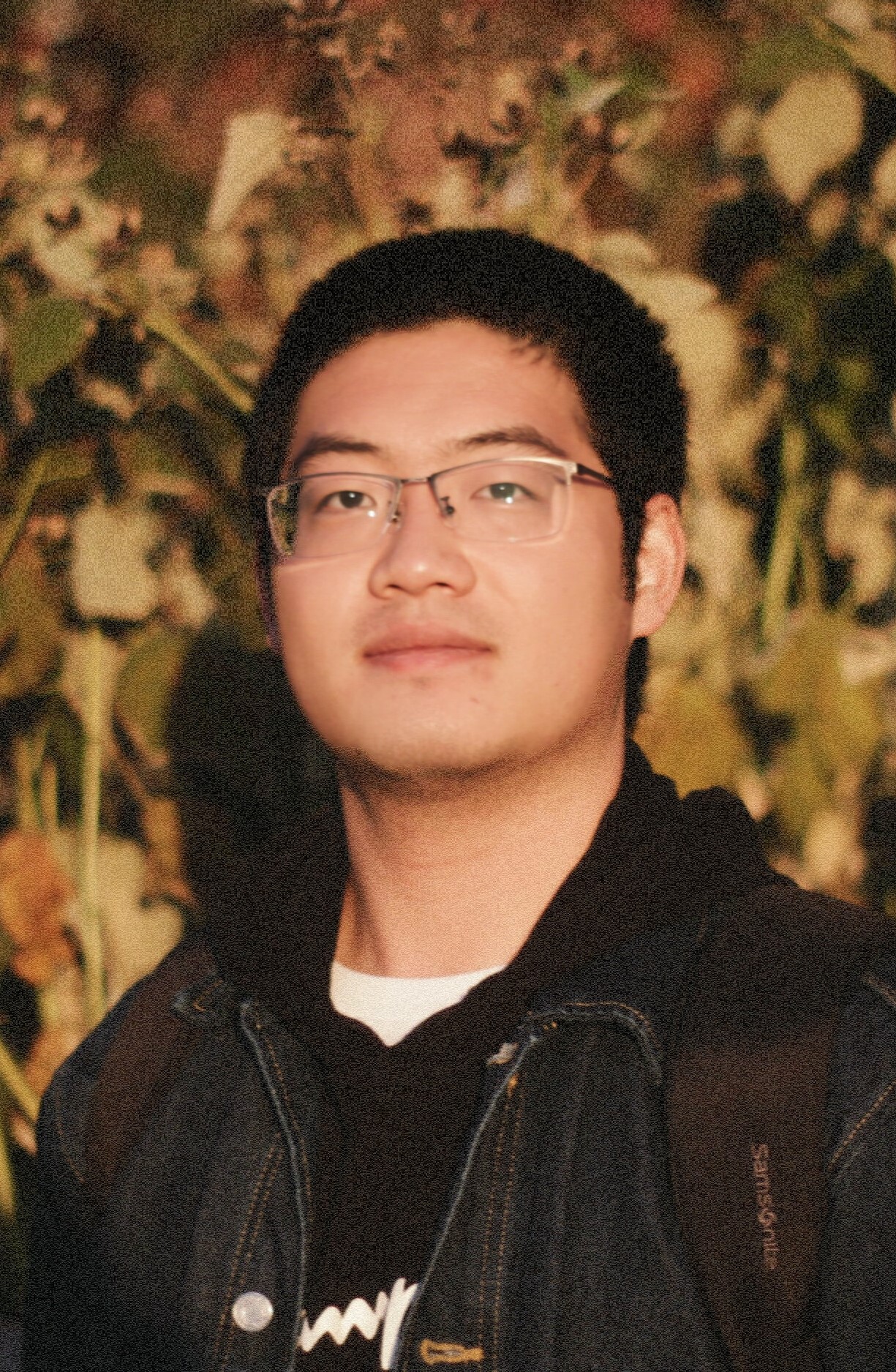}}]{Weicheng Dai} is a Postgraduate Associate at Yale University.  He received the B.S. and M.Sc. degree from Southeast University and New York University, respectively. His research interests include computer vision, medical image analysis, natural language processing, and theoretical machine learning. He has published papers at top-tier conferences in MICCAI and IPMI, etc.  He has served as reviewers for IEEE-TMI, MICCAI, AAAI, and IJCAI.
\end{IEEEbiography}
\vspace{-20pt}

\begin{IEEEbiography}[{\includegraphics[width=1in,height=1.25in,clip,keepaspectratio]{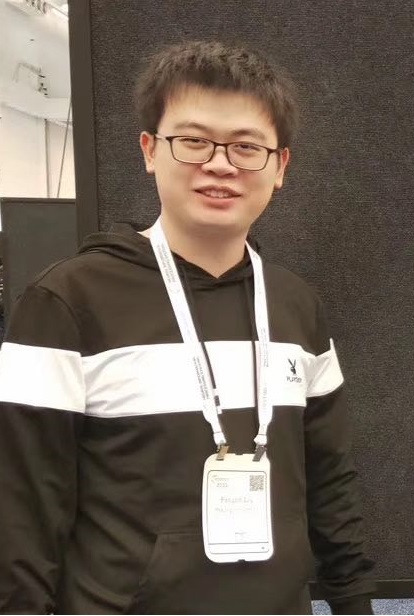}}]{Fenglin Liu} 
% \begin{IEEEbiography}{Fenglin Liu}
is a PhD student at the University of Oxford. His research interests include Natural Language Processing (NLP), especially vision-and-language, Machine Learning, and their applications to clinical, i.e., Clinical NLP.
He has published papers at top-tier journals and conferences, e.g., TPAMI, NeurIPS, CVPR, ACL, EMNLP, NAACL. He has served as a senior program committee member for IJCAI and was awarded as the Distinguished/Outstanding Reviewer of CVPR, AAAI, and IJCAI.
\end{IEEEbiography}
\vspace{-20pt}

\begin{IEEEbiography}[{\includegraphics[width=1in,height=1.25in,clip,keepaspectratio]{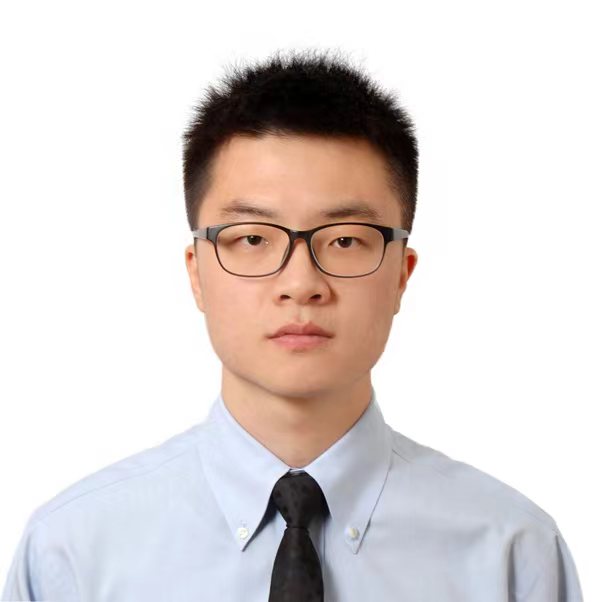}}]{Yifei Min} is a Ph.D. student in Statistics and Data Science at Yale University. He received the B.S. and M.A. degree from the University of Hong Kong and University of Pennsylvania, respectively. He has a broad research interest in the area of machine learning theory with a focus on online learning, reinforcement learning, and medical imaging analysis. Previously, several of his work have been published in NeurIPS, ICLR, ICML, UAI, AISTATS, etc. 
\end{IEEEbiography}
\vspace{-20pt}

\begin{IEEEbiography}[{\includegraphics[width=1in,height=1.25in,clip,keepaspectratio]{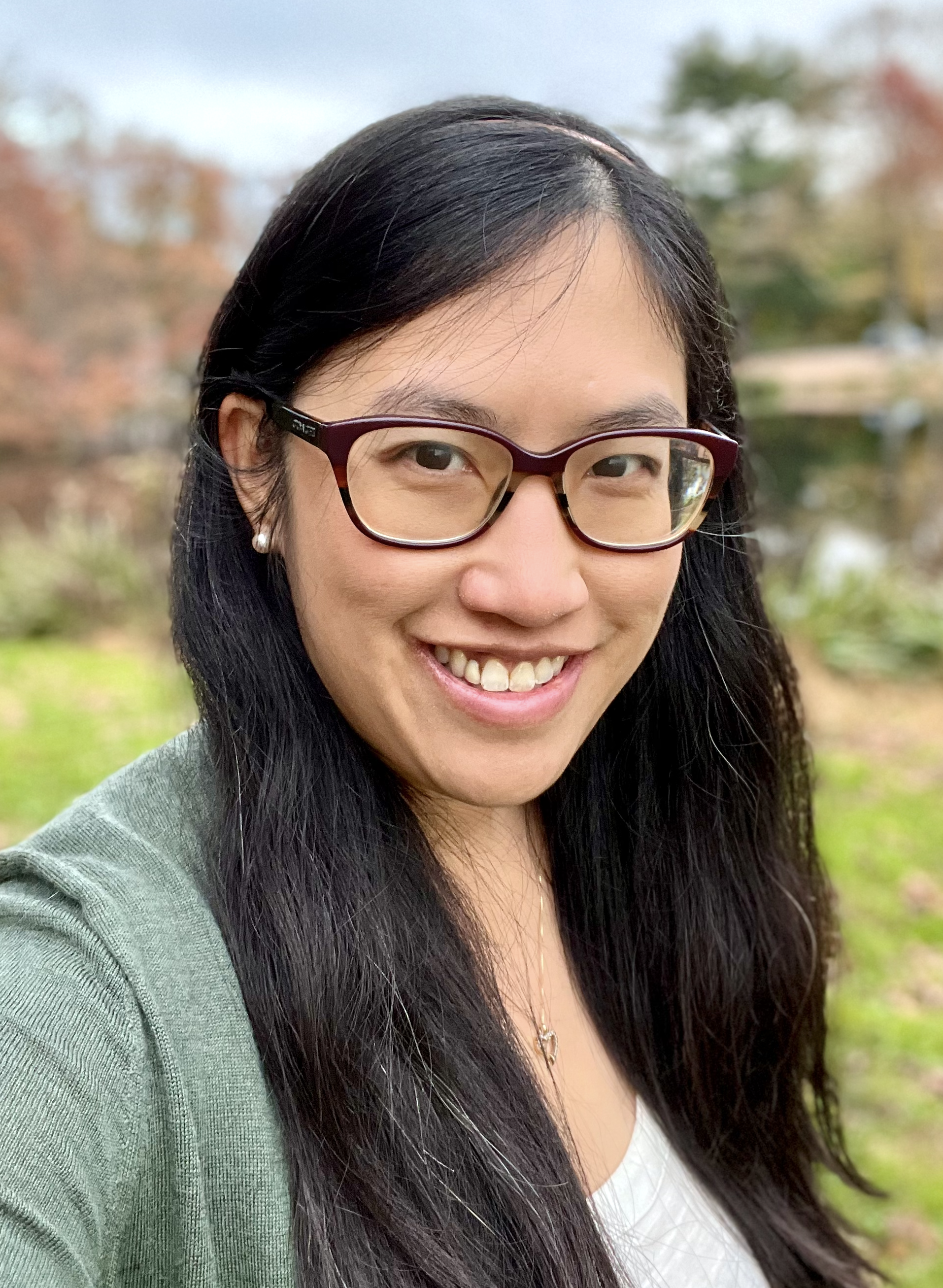}}]{Nicha C. Dvornek} is an Assistant Professor of Radiology \& Biomedical Imaging and Biomedical Engineering at Yale University. She received the B.S. in Biomedical Engineering at Johns Hopkins University and the M.S., M.Phil., and Ph.D. in Engineering and Applied Science at Yale University. She completed postdoctoral training in Diagnostic Radiology at Yale and as a T32 Fellow with the Yale Child Study Center. Her recent work focuses on the development and application of novel machine learning approaches to medical image processing and analysis, with applications spanning from neurodevelopmental disorders to cancer. Dr. Dvornek and her team’s work has been recognized by multiple paper and poster awards at international meetings. She is a member of the MICCAI society and on the board of the Women in MICCAI. She serves as an associate editor for Computerized Medical Imaging and Graphics and the Journal of Medical Imaging, an area chair for MICCAI and MIDL conferences, and a reviewer for multiple journals and conferences such as Medical Image Analysis, IEEE Transactions on Medical Imaging, NeurIPS, and CVPR.
\end{IEEEbiography}
\vspace{-20pt}

\begin{IEEEbiography}[{\includegraphics[width=1in,height=1.25in,clip,keepaspectratio]{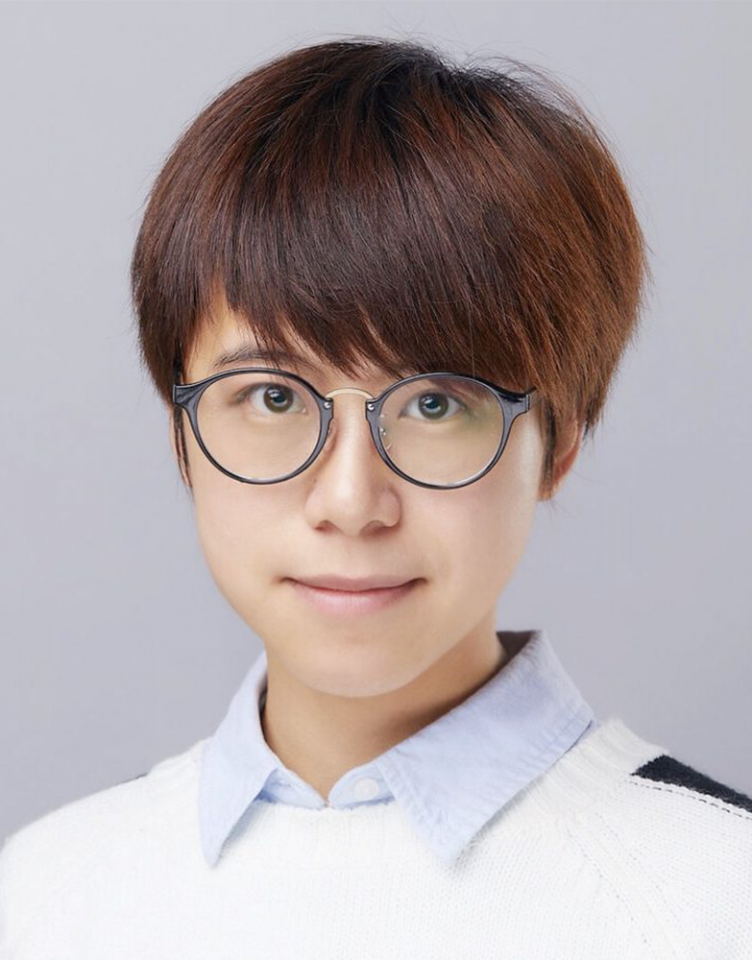}}]{Xiaoxiao Li} 
received the Ph.D. degree from Yale University, New Haven, CT, USA, in 2020. She was a Post-Doctoral Research Fellow with the Computer Science Department, Princeton University, Princeton, NJ, USA. Since August 2021, she has been an Assistant Professor at the Department of Electrical and Computer Engineering (ECE), Uni- versity of British Columbia (UBC), Vancouver, BC, Canada. In the last few years, she has over 30 papers published in leading machine learning conferences and journals, including NeurlPS, ICML, ICLR, MICCAL, IPML, BMVC, IEEE T RANSACTIONS ON M EDICAL LMAGING , and Medical lmage Analysis. Her research interests include range across the interdisciplinary fields of deep learning and biomedical data analysis, aiming to improve the trustworthiness of artificial intelligence (AI) systems for health care. Dr. Li’s work has been recognized with several best paper awards at international conferences.
\end{IEEEbiography}
\vspace{-20pt}

\begin{IEEEbiography} 
[{\includegraphics[width=1in,height=1.25in,clip,keepaspectratio]{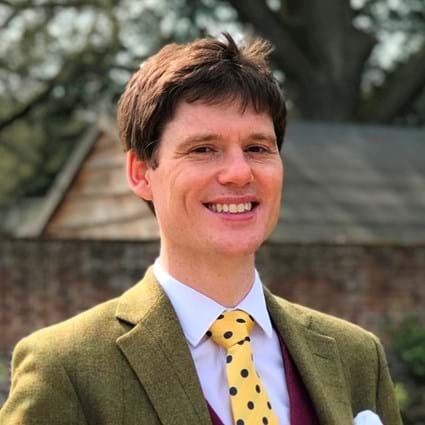}}]{David A. Clifton} is the Royal Academy of Engineering Chair of Clinical Machine Learning at the University of Oxford, and OCC Fellow in AI \& Machine Learning at Reuben College, Oxford.  He was the first AI scientist to be appointed to an NIHR Research Professorship, which is the UK medical research community's ``flagship Chair programme''.  He is a Fellow of the Alan Turing Institute, Research Fellow of the Royal Academy of Engineering, Visiting Chair in AI for Healthcare at the University of Manchester, and a Fellow of Fudan University, China.
He studied Information Engineering at Oxford's Department of Engineering Science, supervised by Prof. Lionel Tarassenko CBE, Chair of Electrical Engineering.  His research focuses on the development of machine learning for tracking the health of complex systems. His previous research resulted in patented systems for jet-engine health monitoring, used with the engines of the Airbus A380, the Boeing 787 ``Dreamliner'', and the Eurofighter Typhoon. Since graduating from his DPhil in 2009, he has focused mostly on the development of AI-based methods for healthcare. Patents arising from this collaborative research have been commercialised via university spin-out companies OBS Medical, Oxehealth, and Sensyne Health, in addition to collaboration with multinational industrial bodies.  He was awarded a Grand Challenge award from the UK Engineering and Physical Sciences Research Council, which is an EPSRC Fellowship that provides long-term strategic support for ``future leaders in healthcare''.  His research has been awarded over 35 academic prizes; in 2018, he was joint winner of the inaugural ``Vice-Chancellor's Innovation Prize'', which identifies the best interdisciplinary research across the entirety of the University of Oxford.
\end{IEEEbiography}
\vspace{-20pt}

\begin{IEEEbiography} 
[{\includegraphics[width=1in,height=1.25in,clip,keepaspectratio]{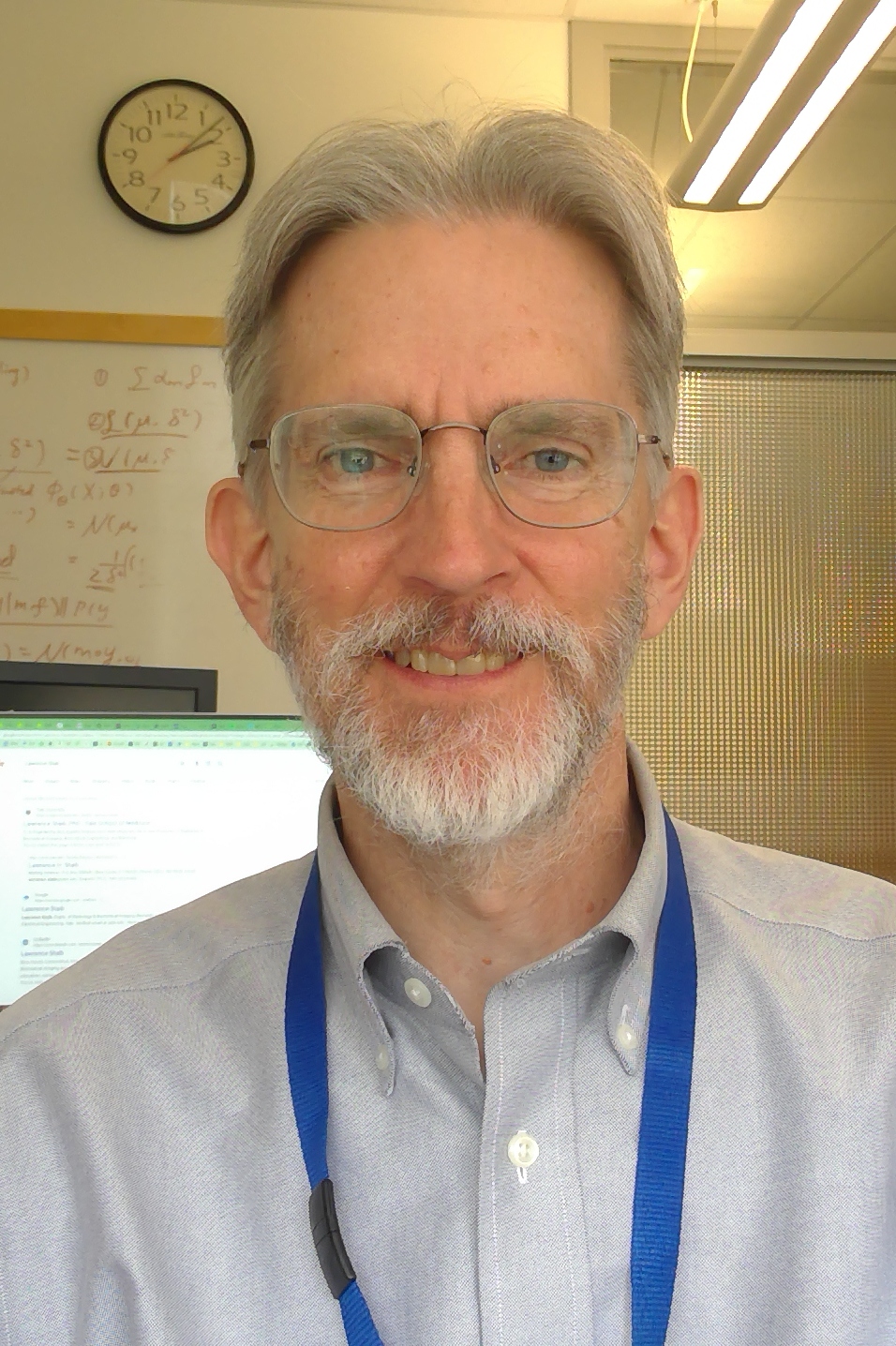}}]{Lawrence Staib} is Professor of Radiology \& Biomedical Imaging, Biomedical Engineering, and Electrical Engineering at Yale University. He is Director of Undergraduate Studies in Biomedical Engineering at Yale. He is a Fellow of the American Institute for Medical and Biological Engineering and of the Medical Image Computing and Computer-Assisted Intervention Society. He is a Distinguished Investigator of the Academy for Radiology \& Biomedical Imaging Research.  He received his B.A. in Physics from Cornell University and his Ph.D. in Engineering and Applied Science from Yale University.  He is a member of the editorial board of Medical Image Analysis and Associate Editor of IEEE Transactions on Biomedical Engineering. His research interests are in medical image analysis and machine learning.
\end{IEEEbiography}
\vspace{-20pt}

\begin{IEEEbiography} 
[{\includegraphics[width=1in,height=1.25in,clip,keepaspectratio]{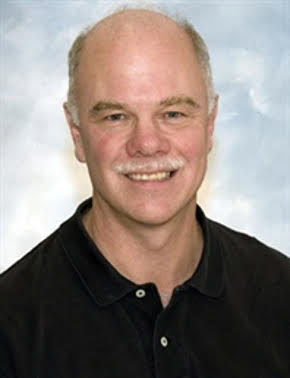}}]{James S. Duncan} is the Ebenezer K. Hunt Professor of Biomedical Engineering and a Professor of Radiology \& Biomedical Engineering, Electrical Engineering and Statistics \& Data Science at Yale University. He is currently the Chair of the Department of Biomedical Engineering. Professor Duncan received his B.S.E.E. with honors from Lafayette College (1973),  his M.S. (1975) in Electrical Engineering from the University of California, Los Angeles and his Ph.D. from the University of Southern California (1982) in Electrical Engineering. Dr. Duncan’s research efforts have been in the areas of computer vision, image processing, and medical imaging, with an emphasis on biomedical image analysis. These efforts have included the segmentation of deformable structure from 3D image data, the tracking of non-rigid motion/deformation from spatiotemporal images, and the development strategies for image-guided intervention/surgery. Most recently, his efforts have focused on the development of image analysis strategies based on the integration of data-driven machine learning and physical model-based approaches. He has published over 325 peer-reviewed articles in these areas and has been the principal investigator on a number of peer-reviewed grants from both the National Institutes of Health and the National Science Foundation over the past 40 years. Professor Duncan is a Fellow of the Institute of Electrical and Electronic Engineers (IEEE), of the American Institute for Medical and Biological Engineering (AIMBE) and of the Medical Image Computing and Computer Assisted Intervention (MICCAI) Society. In 2012, he was elected to the Council of Distinguished Investigators, Academy of Radiology Research.  In 2014 he was elected to the Connecticut Academy of Science \& Engineering. He currently serves as the co-Editor-in-Chief of Medical Image Analysis, and has served on the editorial boards of the IEEE Transactions on Medical Imaging, the  Journal of Mathematical Imaging and Vision and the Proceedings of the IEEE.  He is a past President of the MICCAI Society and in 2017 received the MICCAI Society’s Enduring Impact Award.
\end{IEEEbiography}

\clearpage
\appendices
\section{Theoretical Analysis}
\label{section:ablation-theory}

In this section, we provide a theoretical justification for our MONA.
We will focus on exploring how the student-teacher architecture in MONA functions in a way that helps the model generalize well in the limited-label setting.

To begin with, we introduce some notations and definitions to facilitate our analysis in this section.
We denote an image by $\mathbf{x}$ and its label (segmentation map) by $\mathbf{y}$. 
A segmentation model is a function $f$ such that $f(\mathbf{x})$ is an output segmentation of the image $\mathbf{x}$. 
Let a (supervised) loss function be denoted by $\ell(\cdot)$, such that the loss of a model $f(\cdot)$ for a sample $(\mathbf{x}, \mathbf{y})$ is given as $\ell(f(\mathbf{x}), \mathbf{y})$.  
Here for generality, we do not specify the exact form of $\ell$, and $\ell$ can be taken as cross-entropy loss, DICE loss, etc.

Since MONA is a semi-supervised framework, we start by considering the supervised part.
Specifically, we adopt the typical setting of empirical risk minimization (ERM).
Assume there are $n$ data labeled samples $\mathcal{S}_n = \{\mathbf{x}_i, \mathbf{y}_i\}_{i=1}^n$ from some population distribution $\mathcal{D}$. 
Let $\mathcal{F}$ denote a function class that contains all the candidate models (\eg, parameterized family of neural networks). 
The empirical risk of a model $f$ (\ie supervised training loss) is defined as:
\begin{align*}
    R_{\ell, f} (\mathcal{S}_n)=\frac{1}{n}\sum_{i=1}^n \ell(f(\mathbf{x}_i), \mathbf{y}_i).
\end{align*}
Supervised learning learns a model by ERM:
\begin{align*}
    \hat{f} = \arg\min_{f \in \mathcal{F}} R_{\ell, f}(\mathcal{S}_n).
\end{align*}

\paragraph{Rademacher complexity} 
To understand how ERM in the supervised part of MONA is related to its generalization ability, we invoke the tool of Rademacher complexity. 
For any function class $\mathcal{G}$, its Rademacher complexity is defined as follows:
\begin{definition}[Rademacher complexity]\label{def: rademacher}
    Assume there is a labeled sample set $\mathcal{S}_n = \{\mathbf{x}_i, \mathbf{y}_i\}_{i=1}^n$, such that $(\mathbf{x}_i, \mathbf{y}_i)$ are $i.i.d.$ samples from some population distribution $\mathcal{D}$. Let $\sigma = \{\sigma_i\}_{i=1}^n$ be $i.i.d.$ Bernoulli random variables. 
    The empirical Rademacher complexity $\mathfrak{R}_{\mathcal{S}_n}(\mathcal{G})$ is defined as:
    \begin{align*}
        \mathfrak{R}_{\mathcal{S}_n}(\mathcal{G}) = \frac{1}{n} \mathbb{E}_{\sigma} \left[ \sup_{g \in \mathcal{G}} \sum_{i=1}^n \sigma_i g (\mathbf{x}_i, \mathbf{y}_i) \right].
    \end{align*}
    The Rademacher complexity $\mathfrak{R}_n (\mathcal{G})$ is defined as:
    \begin{align*}
        \mathfrak{R}_n (\mathcal{G}) = \mathbb{E}_{\mathcal{S}_n \sim \mathcal{D}^n}[\mathfrak{R}_{\mathcal{S}_n}(\mathcal{G})].
    \end{align*}
\end{definition}

The model function class $\mathcal{F}$ and the loss function $\ell$ together induce a function class $\mathcal{G}_{\ell, \mathcal{F}}$ defined as:
\begin{align*}
    \mathcal{G}_{\ell, \mathcal{F}} = \left\{ g:(\mathbf{x}, \mathbf{y})\to \ell(f(\mathbf{x}), \mathbf{y}) \ \middle  | \ f \in \mathcal{F}\right\}.
\end{align*}
The power of Rademacher complexity is that it can relate
the generalization ability of the model learned by ERM with the training loss and sample size. 
\begin{theorem}[\cite{bartlett2002rademacher}]
    Let $0 < \delta < 1$. With probability at least $1-\delta$ over the distribution of the sample set $\mathcal{S}_n$, for all $g \in \mathcal{G}_{\ell,\mathcal{F}}$ such that $g(\mathbf{x},\mathbf{y}) = \ell(f(\mathbf{x}),\mathbf{y})$, it holds that:
    \begin{align*}
        \mathbb{E}_{\mathcal{D}} [g(\mathbf{x},\mathbf{y})] \leq R_{\ell, f} (\mathcal{S}_n) + 2 \mathfrak{R}_n(\mathcal{G}_{\ell, \mathcal{F}}) + \sqrt{\frac{\log(1/\delta)}{n}} . 
    \end{align*}
\end{theorem}

\begin{remark}\label{remark: interpret rademacher 1}
    The left-hand side of the above result is essentially $\mathbb{E}_{\mathcal{D}} [g(\mathbf{x},\mathbf{y})] = \mathbb{E}_{\mathcal{D}} [\ell(f(\mathbf{x}),\mathbf{y})]$, which is the generalization error of the model $f$.
    Therefore, the above theorem provides an upper bound of the generalization error, which depends on the empirical risk (\ie supervised training loss) $R_{\ell,f}(\mathcal{S}_n)$, the labeled sample size $n$, and the Rademacher complexity $\mathfrak{R}_n(\mathcal{G}_{\ell, \mathcal{F}})$. 
\end{remark}

By Remark~\ref{remark: interpret rademacher 1}, to learn a model with small generalization error, all three terms in the right-hand side need to be made small. 
The supervised training aims at minimizing the empirical risk $R_{\ell,f}(\mathcal{S}_n)$. 
The last term, $\sqrt{\frac{\log(1/\delta)}{n}}$, is decreasing in the labeled sample size $n$.
For medical imaging tasks where labeled samples are limited, the sample size $n$ is a bottleneck and usually cannot be very large.
Therefore, one approach is to reduce the Rademacher complexity of the function class $\mathcal{G}_{\ell, \mathcal{F}}$. 
In the following, we show that the student-teacher architecture of MONA can implicitly reduce the Rademacher complexity of the function class $\mathcal{G}_{\ell, \mathcal{F}}$ by restricting the complexity of the model function class $\mathcal{F}$.
Due to the complex nature of the segmentation task, we consider some common setting where the complexity of the model class $\mathcal{F}$ is determined by certain basis functions.

\paragraph{Function class with basis}

We consider  the model function class $\mathcal{F}$ which contains a finite set of basis functions.
\begin{definition}[Finite-basis function class]\label{definition: basis function}
A real-valued function class $\mathcal{F}$ is a finite-basis function class if it satisfies the following:
There exists a finite set of functions $\{\phi_i\}_{i=1}^m$, such that for any $f\in\mathcal{F}$, there exists a coordinate vector $\{c_i\}_{i=1}^m$ such that $f (\cdot) = \sum_{i=1}^m c_i \phi_i (\cdot)$.
 % Without loss of generality we assume $|c_i|\leq 1$ for all $i$. 
\end{definition}

A finite-basis function class has a desired property which allows us to bound its Rademacher complexity.
For example, the following lemma gives the upper bound of the Rademacher complexity of when the underlying basis are linear functions. 

\begin{lemma}\label{lem: rademacher of finite basis}
    Suppose $\mathcal{F}$ is a finite-basis function class with linear basis functions $\{\phi_i\}_{i=1}^n$ such that $\|\phi_i\|_{\infty} \leq V_i$ for some $V_i > 0$, for all $i=1,\cdots, m$. 
    Assume that for any $f\in\mathcal{F}$, $f = \sum_{i=1}^m c_i \phi_i$ where $|c_i| < C_i$ for some fixed $C_i > 0 $. 
    Then the Rademacher complexity of $\mathcal{F}$ is upper bounded by:
    \begin{align*}
        \mathfrak{R}_n (\mathcal{F}) \leq \Tilde{\mathcal{O}}\left((\sum_{j=1}^m C_j ) \cdot \max_{j\in[m]} V_j \cdot \frac{1}{\sqrt{n}} \right),
    \end{align*}where $\tilde{\mathcal{O}}$ hides logarithmic and constant factors. 
\end{lemma}

\begin{proof}[Proof of Lemma~\ref{lem: rademacher of finite basis}]
    % For a single function $\phi_i$, we have 
    % \begin{align*}
    %     \mathfrak{R}_{\mathcal{S}_n}( \{\phi_i\} ) = \frac{1}{n} \mathbb{E}_{\sigma} \left[ \sup_{g \in \mathcal{G}} \sum_{i=1}^n \sigma_i g (\mathbf{x}_i, \mathbf{y}_i) \right] \leq V_i,
    % \end{align*} since $|\sigma_i| \leq 1$. This implies $\mathfrak{R}_{n}( \{\phi_i\} ) \leq V_i$ according to Definition~\ref{def: rademacher}.
    First, for a linear function $\phi_j(\cdot)$, its Rademacher complexity is $\mathcal{O}(V_j \frac{1}{\sqrt{n}})$ \cite{bartlett2002rademacher}.
    Furthermore, if we define a function class as $\mathcal{H}_j = \{h(\cdot) = c_j \phi_j(\cdot) : |c_j| \leq C_j\}$, then $\mathcal{H}_j$ is again a linear function class, which contains functions with absolute value upper bounded by $C_j V_j$. 
    This implies that
    \begin{align*}
        \mathfrak{R}_n (\mathcal{H}_j) = \Tilde{\mathcal{O}}\left(C_j \cdot V_j \cdot \frac{1}{\sqrt{n}} \right)
    \end{align*}
    Finally, we invoke Theorem~12 in \cite{bartlett2002rademacher}, 
    which shows the subadditivity of Rademacher complexity, \ie, for any function classes $\mathcal{V}$ and $\mathcal{V}'$, we have $\mathfrak{R}_n (\mathcal{V} + \mathcal{V}') \leq \mathfrak{R}_n (\mathcal{V}) + \mathfrak{R}_n(\mathcal{V}')$.
    Applying this to $\mathcal{H}_1 + \cdots + \mathcal{H}_m$ and then using H\"older's inequality over $\sum_{j=1}^m C_j V_j$, we finish the proof.
    % By taking $h$ as $\phi_i$ and using the fact that the Rademacher complexity of a zero function is zero, we get that:
    % \begin{align*}
    %     \mathfrak{R}_{n}( \{c\phi_i: c \in [0, 1]\} ) \leq V_i \cdot \frac{1}{\sqrt{n}}.
    % \end{align*}
    % The rest of the proof follows directly from the sub-linearity of Rademacher complexity. 
    % Specifically, by Definition~\ref{definition: basis function}, 
    % it holds that $f (\cdot) = \sum_{j=1}^m c_j \phi_j (\cdot)$, and therefore
    % \begin{align*}
    %     & \mathfrak{R}_{\mathcal{S}_n} (\mathcal{F})
    %     \\ & = \frac{1}{n} \mathbb{E}_{\sigma} \left[ \sup_{f \in \mathcal{F}} \sum_{i=1}^n \sigma_i f (\mathbf{z}_i) \right]
    %     \\ & = \frac{1}{n} \mathbb{E}_{\sigma} \left[ \sup_{\{c_j\}_{j=1}^m} \sum_{i=1}^n \sigma_i \left(\sum_{j=1}^m c_j \phi_j (\mathbf{z}_i) \right) \right] 
    %     \\ & = \frac{1}{n} \mathbb{E}_{\sigma} \left[ \sup_{\{c_j\}_{j=1}^m} \sum_{j=1}^m c_j \sum_{i=1}^n \sigma_i  \phi_j (\mathbf{z}_i) \right].
    % \end{align*}
    % Since $|c_j| \leq C_j$ for all $j=1,\cdots,m$, it follows that 
    % \begin{align*}
    %     &\mathfrak{R}_{\mathcal{S}_n} (\mathcal{F}) 
    %     \\ & \leq \frac{1}{n} \mathbb{E}_{\sigma} \left[ \sum_{j=1}^m \sup_{\{c_j\}_{j=1}^m}  c_j \sum_{i=1}^n \sigma_i  \phi_j (\mathbf{z}_i) \right]
    %     \\ & = \sum_{j=1}^m \frac{1}{n} \mathbb{E}_{\sigma} \left[ \sup_{c_j}  c_j \sum_{i=1}^n \sigma_i  \phi_j (\mathbf{z}_i) \right]
    %     \\ & \leq \sum_{j=1}^m \frac{1}{n} C_j \mathbb{E}_{\sigma} \left[  \sum_{i=1}^n \sigma_i  \phi_j (\mathbf{z}_i) \right].
    % \end{align*}
\end{proof}

With the finite-basis function class defined, we now consider the student-teacher architecture of MONA and show that it can help reduce the effective size of the basis, thus reducing the Rademacher complexity in certain examples.
To this end, we point out that the student-teacher architecture in MONA is a modified version of self-distillation \cite{furlanello2018born, zhang2020self}.
Specifically, self-distillation refers to the case where the \textit{student} model and the \textit{teacher} model share the same architecture, and the predictions of the trained student model are fed back in as new target values for retraining iteratively. 
Specifically, denote by $f_t$ the model at the $t$-th iteration, the simplest form of self-distillation is called one-step self-distillation, which updates the model at the $(t+1)$-st iteration by 
$f_{t+1} = \arg\min_{f\in\mathcal{F}} \sum_{i=1}^n \ell(f(\mathbf{x}_i, \mathbf{y}_{i,t}))$, where $\mathbf{y}_{i,t} = f_{t}(\mathbf{x}_i)$ is the pseudo-label generated by $f_t$. 
The student-teacher architecture in MONA is multi-step generalization of self-distillation.
In MONA, the teacher model is set to be the EMA of the student, and its features are used to implement contrastive learning for the student model. 
In other words, if the student model at the $(t+1)$-st iteration is denoted by $f_{t+1}$, then the teacher model can be denoted as $\sum_{j=0}^w \gamma_j f_{t-j}$ by the definition of EMA. 
Furthermore, the teacher model of MONA is used to generate positive/negative examples for contrastive loss. This is a bit different from the pseudo-labels in the vanilla self-distillation.
Although this is not exactly the same as $f_t$, the learning dynamic is similar. 
This explains that the student-teacher architecture in MONA is indeed self-distillation.

We now demonstrate how self-distillation can help reduce the size of the basis.
Since self-distillation and EMA-based self-distillation are extremely complicated to analyze directly, there is a lack of theoretical study. 
One notable exception is \cite{mobahi2020self}, which considers a simple example of kernel-regularized regression.
Therefore, we follow the framework of \cite{mobahi2020self} and show that under their proposed example, self-distillation can reduce the size of the basis. 
Specifically, consider the following problem:
\begin{align}\label{eq: regularized problem}
    f^{*} & = \arg\min_{f \in \mathcal{F}} \int_{\mathbf{x}} \int_{\mathbf{x}'} \kappa(\mathbf{x}, \mathbf{x}') f(x) f(x') d \mathbf{x} d \mathbf{x}' , \notag
    \\  &\quad  s.t.  \ \ \frac{1}{n} \sum_{i=1}^n \| f(\mathbf{x}_i) - \mathbf{y}_i \|_2^2 \leq \epsilon,
\end{align}where $\kappa(\cdot,\cdot)$ is a positive definite kernel function. 
The above is a simple regularized regression problem, where the loss function $\ell$ is the $\ell_2$ loss.
% Although the $\ell_2$ loss differs from the contrastive loss of MONA, it is actually closely related as shown by recent work.

For the problem described above, it can be shown that the solution via self-distillation has a closed-form solution.
\begin{proposition}[\cite{mobahi2020self}]\label{prop: form of ft}
    For the problem defined by Eqn.\ref{eq: regularized problem}, suppose we solve it using one-step self-distillation.
    Then there exists matrix-valued function $G(\mathbf{x})$, matrices $\mathbf{V}$, $\mathbf{D}$, and $\{\mathbf{A}_\tau\}_{\tau=0}^t$. Here we will not introduce the exact form of these matrices here since this not our focus. We refer readers to section 2.2 of \cite{mobahi2020self}. For our purpose, we will only discuss the property of the matrices $\{\mathbf{A}_\tau\}_{\tau=0}^t$.,
    such that:
\begin{align}\label{eq: ft}
    f_t (\mathbf{x}) = \mathbf{G}(\mathbf{x})^\top \mathbf{V}^\top \mathbf{D}^{-1} (\Pi_{\tau=0}^t \mathbf{A}_\tau) \mathbf{V} \mathbf{Y}_0, 
\end{align} where $\mathbf{Y}_0$ is the unknown vector of ground truth labels for $\mathbf{x}_1, \cdots, \mathbf{x}_n$.
\end{proposition}

We point out that, from Proposition~\ref{prop: form of ft}, the only $t$-dependent term is the matrices product $\Pi_{\tau=0}^t \mathbf{A}_\tau$. 
We denote $\mathbf{B}_t = \Pi_{\tau=0}^t \mathbf{A}_\tau$.
Importantly, it has been shown that $\mathbf{B}_t$ becomes sparse over time.
\begin{theorem}[Theorem 5, \cite{mobahi2020self}]\label{thm: shrinking}
    Under the same assumption as Proposition~\ref{prop: form of ft}, $\mathbf{B}_t$ is a diagonal matrix. 
    Furthermore, the diagonal element of $B_t$ satisfies that, for any $j,l$,
    \begin{align*}
        \frac{\mathbf{B}_{t}(j,j)}{\mathbf{B}_{t}(l,l)} \geq \left( \frac{\frac{\|\mathbf{Y}_0\|}{\sqrt{n \epsilon}}-1+\frac{d_{\min}}{d_j}}{\frac{\|\mathbf{Y}_0\|}{\sqrt{n \epsilon}}-1+\frac{d_{\min}}{d_l}} \right)^{t+1},
    \end{align*} where $d_j, d_l>0$ are diagonal elements of the matrix $\mathbf{D}$.
\end{theorem}

Theorem~\ref{thm: shrinking} implies the basis reducing the effect of self-distillation. %\chenyu{Not sure}
Specifically, assume without loss of generality that $d_j<d_l$, Theorem \ref{thm: shrinking} implies that $\mathbf{B}_{t}(l,l)\leq \rho_{k,l}^{-t-1} \mathbf{B}_{t}(j,j) $, where $\rho_{k,l} < 1$. 
Therefore, all the small diagonal elements of $\mathbf{B}_t$ will be negligible compared to the largest element. 
In other words, $\mathbf{B}_t$ is a soft-sparse matrix.
Furthermore, from Eqn.~\ref{eq: ft}, we can see that $f_t (\mathbf{x})$ can be viewed as:
\begin{align*}
    f_t (\mathbf{x}) = [\mathbf{G}(\mathbf{x})^\top \mathbf{V}^\top \mathbf{D}^{-1}] [\mathbf{B}_t \mathbf{V} \mathbf{Y}_0], 
\end{align*}where $[\mathbf{B}_t \mathbf{V} \mathbf{Y}_0]$ is a column vector with sparsity (since $\mathbf{B}_t$ is a sparse diagonal matrix). 
Therefore, we can indeed view $f_t (\mathbf{x})$ as a function of finite basis with sparse coordinate vector $[\mathbf{B}_t \mathbf{V} \mathbf{Y}_0]$. 
Specifically, denote the vector $\mathbf{V} \mathbf{Y}_0$ by $\mathbf{u} = [\mathbf{u}(j)]_{j=1}^m$.
Then the vector $[\mathbf{B}_t \mathbf{V} \mathbf{Y}_0]$ can be written as:
\begin{align*}
    [\mathbf{B}_t \mathbf{V} \mathbf{Y}_0] = [\mathbf{B}_t(j,j) \mathbf{u}(j)]_{j=1}^m.
\end{align*}
Since this vector $[\mathbf{B}_t(j,j) \mathbf{u}(j)]_{j=1}^m$ does not involve $\mathbf{x}$, it can be viewed as the coordinate vector, and the entries of $[\mathbf{G}(\mathbf{x})^\top \mathbf{V}^\top \mathbf{D}^{-1}]$ can be viewed as basis functions. 
% By Lemma~\ref{lem: rademacher of finite basis}, 
If the coordinate vector is sparse, then according to Lemma~\ref{lem: rademacher of finite basis}, 
most of $C_i$'s will be negligible, and thus $\sum_{j=1}^m C_j$ becomes small. 
This indicates that, if we view $\mathcal{F}$ as the training-induced model function class, then its Rademacher compleity $\mathfrak{R}_n (\mathcal{F})$ will be actually be small compared to the entire model function class. 
Finally, since the induced function class $\mathcal{G}_{\ell,f}$ is a composition of $\ell$ and $f$, the Rademacher complexity of $\mathcal{G}_{\ell,f}$ is at most  $\mathcal{O} (L_{\ell} \mathfrak{R}_n (\mathcal{F}))$, where $L_{\ell}$ is the Lipschitz constant of $\ell(\cdot, \mathbf{y})$.
Therefore we see that self-distillation (or the student-teacher architecture) can effectively reduce the complexity of the function class and lead to better generalization performance. 

% Specifically, denote the vector $\mathbf{V} \mathbf{Y}_0$ by $\mathbf{u} = [\mathbf{u}(j)]_{j=1}^m$.
% Then the vector $[\mathbf{B}_t \mathbf{V} \mathbf{Y}_0]$ can be written as
% \begin{align*}
%     [\mathbf{B}_t \mathbf{V} \mathbf{Y}_0] = [\mathbf{B}_t(j,j) \mathbf{u}(j)]_{j=1}^m.
% \end{align*}

% To understand why the self-distillation architecture can implicitly reduce the Rademacher complexity of the model family, we consider the following example. 

\end{document}